\documentclass[aps,nofootinbib,twocolumn,showkeys,floatfix]{revtex4}
\usepackage[utf8]{inputenc}
\usepackage[brazilian]{babel}
\usepackage[T1]{fontenc}
\usepackage{mathrsfs}
\usepackage{amsmath}
\usepackage{amsthm}
\usepackage{hyperref}
\usepackage{amssymb}
\usepackage{graphicx}
\usepackage[dvipsnames]{xcolor}
\usepackage{csquotes}
\usepackage{booktabs}
\usepackage{cancel}
\usepackage{comment}
\usepackage{setspace}

\newtheorem{definition}{Definição}
\newtheorem{theorem}{Teorema}

\newtheorem{corollary}{Corolário}
\newtheorem{ppdd}{Propriedade}

\newtheorem{proposition}{Proposição}[section]

\begin{document}

\title{\Large Revisitando Entropias: Propriedades Formais e Conexões Entre as Entropias de Boltzmann-Gibbs, Tsallis e Rényi\\[1ex] \small Revisiting Entropies: Formal Properties and Connections Between Boltzmann-Gibbs, Tsallis and Rényi Entropies}

\author{Kelvin dos Santos Alves \href{https://orcid.org/0000-0002-3822-9818}{\includegraphics[scale=0.04]{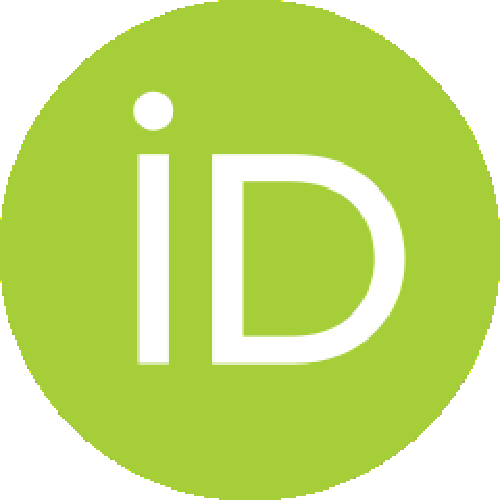}} \thanks{Endereço de correspondência: kelvin.santos@unesp.br}$^{1}$  Rogério Teixeira Cavalcanti\href{https://orcid.org/0000-0001-7848-5472}{\includegraphics[scale=0.04]{orcidicon.eps}} \thanks{Endereço de correspondência: rogerio.cavalcanti@ime.uerj.br}$^{1,2}$ }
\affiliation{$^1$Departamento de Física, Universidade Estadual Paulista, UNESP, Guaratinguetá, SP, Brasil. \\ $^2$Instituto de Matemática e Estatística, Universidade do Estado do Rio de Janeiro, UERJ, Rio de Janeiro, RJ, Brasil.}

\begin{abstract}
\textbf{Resumo}: O objetivo do presente artigo é apresentar uma discussão cuidadosa e acessível sobre aspectos formais das entropias de Boltzmann-Gibbs e de Tsallis. Iniciamos com uma breve exposição da entropia de Boltzmann-Gibbs, abordando suas principais propriedades e os teoremas de unicidade formulados por Shannon e Khinchin. Estabelecidos esses fundamentos, introduzimos a chamada mecânica estatística não aditiva, definindo a entropia de Tsallis, suas propriedades, seu teorema de unicidade e os contrastes com os resultados da mecânica estatística aditiva. Evidenciamos ainda que, em um limite apropriado, os resultados de Boltzmann-Gibbs são recuperados. Concluímos com uma breve discussão sobre a entropia de Rényi e suas conexões com as definições anteriores.   

\vspace{0.5cm}

\textbf{Abstract}: The aim of the present paper is to present a careful and accessible discussion of the formal aspects of Boltzmann-Gibbs and Tsallis entropies. We begin with a brief overview of Boltzmann-Gibbs entropy, highlighting its main properties and the uniqueness theorems formulated by Shannon and Khinchin. Once these foundational results are established, we introduce the framework of nonadditive statistical mechanics, defining Tsallis entropy, discussing its properties and uniqueness theorem, and contrasting it with the results from additive statistical mechanics. We also show that, in an appropriate limit, the Boltzmann-Gibbs results are recovered. The article concludes with a brief discussion of Rényi entropy and its connections to the previously defined entropic forms.
\end{abstract}


\maketitle

\section{Introdução}
\label{intro}

A mecânica estatística usual, fundamentada na entropia de Boltzmann-Gibbs, constitui --- ao lado da Mecânica Clássica e Quântica, do Eletromagnetismo e da Relatividade --- um dos pilares centrais da física contemporânea. Além de ainda ser objeto de intensa pesquisa \cite{cimini2019statistical,lin2024measuring,de2023semiclassical}. Já o conceito de entropia, no contexto da termodinâmica \cite{ClausisusOrig}, apesar de ter sido inicialmente introduzido no século XIX por Clausius, passou a ocupar um papel central na formulação da física estatística apenas com os trabalhos de Boltzmann \cite{boltzmannweitere,boltzmann1877beziehung} e Gibbs \cite{gibbs1902elementary}.

A mecânica estatística de Boltzmann-Gibbs baseia-se na propriedade de aditividade da entropia, tal como definida por esses autores, para sistemas probabilisticamente independentes. Nessa formulação, a entropia é assumida como um postulado fundamental da teoria, uma vez que não pode ser derivada a partir de primeiros princípios, nem resulta de uma dinâmica microscópica específica \cite{tsallis2003introduction}. A formalização matemática desse conceito ganhou maior rigor com os trabalhos de Shannon \cite{book:320619}, em 1948, e Khinchin \cite{book:111623}, em 1953, nos quais foi proposto um conjunto de axiomas para a forma funcional da entropia. A definição de Boltzmann-Gibbs é a única que os satisfaz integralmente.

Uma proposta alternativa de entropia foi feita em 1988 por Tsallis \cite{Tsallis:1987eu}. \textcolor{black}{Essa proposta foi desenvolvida com o objetivo de abranger sistemas físicos para os quais a mecânica estatística de Boltzmann-Gibbs apresentava dificuldades matemáticas ou simplesmente falhava, como é o caso dos chamados sistemas complexos \cite{tsallis2002nonextensive}.} A entropia de Tsallis, usualmente denotada por $S_{q}$, é parte fundamental da denominada mecânica estatística não aditiva\footnote{Popularizada como mecânica estatística não extensiva \cite{tsallis2004nonextensive}.}. O índice $q$ representa um parâmetro associado a características particulares do sistema físico. Essa definição é uma generalização da definição de entropia da mecânica estatística aditiva, que possui a entropia de Boltzmann-Gibbs como um caso particular e parece descrever certos sistemas físicos com bastante precisão \cite{vignat2012quantum,nunes2002nonextensive,tsallis_2021}. Notadamente, uma classe de sistemas com correlações não locais. O parâmetro $q$ está relacionado à não localidade desses sistemas e pode, a princípio, ser determinado a partir da dinâmica microscópica do sistema \cite{tsallis2003introduction}. A propriedade escolhida para ser generalizada é a aditividade, por vezes denominada de pseudo-aditividade.

Outras generalizações do funcional entropia foram propostas, como a entropia de Rényi \cite{renyi1970probability}, denotada por $S_{q}^{R}$. No entanto, o fato de não ser experimentalmente robusta e nem côncava para $q> 1$ faz com que a única proposta que pareça adequada para as abordagens termodinâmicas seja $S_{q}$ \cite{boon2005special}.

O presente trabalho pretende explorar as principais propriedades matemáticas que caracterizam a função entropia, que se apresenta como um postulado fundamental tanto na mecânica estatística de Boltzmann-Gibbs quanto na de Tsallis. Para isso, alguns fundamentos sobre os quais a mecânica estatística é construída são revisitados. Essa discussão não é encontrada facilmente em referências canônicas \cite{book:18204,book:1322644}. Ressaltamos que o objetivo aqui não é discutir o conceito de entropia especificamente no contexto da termodinâmica, da mecânica estatística ou da teoria da informação, mas sim investigar suas propriedades gerais enquanto conceito matemático fundamental, independentemente da área de aplicação. Contudo, por simplicidade e clareza, os exemplos ilustrativos serão apresentados no contexto da teoria da informação. Para tal, primeiro abordaremos a entropia de Boltzmann-Gibbs e as principais propriedades desse funcional entrópico. Após isso, focaremos na entropia de Tsallis, suas principais propriedades, diferenças e semelhanças com respeito à mecânica estatística aditiva, além de uma breve introdução sobre a entropia de Rényi. Ao leitor interessado em uma discussão introdutória sobre aplicações físicas da mecânica estatística não aditiva, indicamos a referência \cite{tsallis_2021}. \textcolor{black}{Na Seção \ref{aplicacoes} damos um breve panorama de tais aplicações, bem como um conjunto mais abrangente de referências sobre o tópico.}

\section{Mecânica estatística de Boltzmann-Gibbs}\label{def_entropia_de_bolt_gibbs_shannon}

Com o surgimento e avanço da Revolução Industrial no século XIX, tornou-se essencial otimizar os processos industriais e aumentar a eficiência das máquinas térmicas a vapor que estavam sendo desenvolvidas. Tais máquinas, como os trens a vapor, operam entre dois reservatórios térmicos: extraem parte do calor de uma fonte quente para realizar trabalho e dissipam o restante em uma fonte fria. Nesse contexto, nasce a Termodinâmica e inserem-se os trabalhos do engenheiro Sadi Carnot, cujo objetivo era compreender como maximizar o rendimento das máquinas térmicas, tornando-as o mais eficientes possível. A gênese do conceito de entropia tem origem nesse período \cite{book:3209553}, em um momento em que o interesse por tais desenvolvimentos transcendeu a esfera científica, envolvendo também aspectos político-econômicos.

Posteriormente aos trabalhos de Carnot, Clausius \cite{CarnotOrig}\cite{ClausisusOrig} e William Thomson \cite{thomson1853xv} elaboraram, distintamente, mas de forma equivalente, a segunda lei da termodinâmica. Os estudos de Clausius que culminaram na segunda lei da termodinâmica foram focados em analisar a reversibilidade e irreversibilidade de máquinas térmicas atuando entre dois reservatórios térmicos. Foi notado por ele que a variação total de uma determinada grandeza, denominada entropia, é nula para processos reversíveis e maior que zero se o sistema for termicamente isolado \cite{book:3209553}. Assim posto, a entropia surge inicialmente como uma medida da irreversibilidade de um sistema termodinâmico \cite{nascimento_prudente_2016}.

Entre o final do século XIX e o início do século XX, Ludwig Boltzmann introduziu uma interpretação estatística da termodinâmica, partindo da ideia de que o calor é resultado do movimento das moléculas \cite{BoltzmannTrab}, com base na teoria cinética dos gases de Maxwell. Para ilustrar sua abordagem, considere o exemplo de duas moedas idênticas: cada uma possui duas possibilidades em um lançamento, cara ou coroa, resultando em quatro possíveis configurações ao se lançar ambas — os microestados. Cada microestado representa uma configuração acessível do sistema. A formulação de Boltzmann leva ao teorema fundamental segundo o qual, quanto maior o número de microestados associados a um estado macroscópico, maior a probabilidade de ocorrência desse estado \cite{reis_bassi_2012}. Isso conduz à famosa expressão da entropia de Boltzmann:
	\begin{align}\label{entropia_BG_estados_igual}
		S_{B}(W) = k_{B}\ln W,
	\end{align}
onde $W$ é o número de microestados acessíveis ao sistema e $k_{B}$ é a constante de Boltzmann. Nesse contexto, a entropia passa a ser interpretada como uma medida da quantidade de configurações possíveis de um sistema físico — ou seja, assume um caráter eminentemente probabilístico.

Mais adiante, Gibbs introduz o conceito de \textit{ensemble} para representar um conjunto de microestados acessíveis a um determinado sistema. Por meio desse formalismo, foi introduzida a hoje denominada entropia de Boltzmann-Gibbs
	\begin{align}\label{entropia_BG_estados_discretos}
		S_{BG}(\{p_{i}\}) = -k \sum^{W}_{i=1} p_{i}\ln p_{i},
	\end{align}	
	\begin{align}\label{condicao_de_normalizacao}
		\sum_{i=1}^{W} p_{i} = 1.
	\end{align}
Sendo $W$, novamente, o número de estados microscópicos discretos e $p_{i}$ a probabilidade do sistema acessar um determinado microestado $i$.{ A probabilidade obedece à condição de normalização fornecida pela equação \eqref{condicao_de_normalizacao}. A constante $k$ é interpretada como sendo a constante de Boltzmann-Gibbs do ponto de vista termodinâmico e, normalmente, $k=1$ é utilizado na teoria da informação \cite{tsallis_2021} \footnote{Mesmo em termodinâmica podemos adotar um sistema de unidades tal que a constante $k$ seja unitária, sem perda de generalidade \cite{tsallis2003introduction}.}.} { Pode-se notar que para um conjunto de estados equiprováveis, ou seja, um conjunto no qual todos os microestados são acessíveis com mesma probabilidade, essa probabilidade é dada por $p_{i} = 1/W$. Assim, o funcional dado na equação \eqref{entropia_BG_estados_discretos} torna-se a entropia de Boltzmann, dada pela equação \eqref{entropia_BG_estados_igual}}.

{Já durante o século XX, o engenheiro e matemático Claude Shannon aplicou a teoria de probabilidades ao estudo da comunicação e da informação em seu artigo seminal \textit{The Mathematical Theory of Communication} \cite{book:320619}. Apesar do contexto aparentemente desconectado da termodinâmica e da física estatística, Shannon obtém uma expressão para a medida da incerteza da taxa de informação produzida que é formalmente idêntica à entropia de Boltzmann-Gibbs \cite{pena2021comunicaccao}. {Em um diálogo com o matemático John von Neumann, Shannon relata que estava indeciso sobre como deveria chamar a expressão encontrada por ele e sugere denominá-la de \textit{informação}. Mas justifica que a palavra já é bastante usada e, então, propõe chamá-la de \textit{incerteza}. Em resposta, Neumann sugere chamá-la de entropia, uma vez que a função incerteza encontrada por Shannon já estava sendo usada em mecânica estatística \cite{tribus1971energy}. De fato, a entropia de Boltzmann-Gibbs revela-se uma expressão fundamental que emerge em contextos diversos, o que ainda suscita debates sobre seu real significado\footnote{Uma discussão sobre a interpretação do conceito de entropia em teoria da informação e em mecânica estatística pode ser encontrado em \cite{camargo02}, bem como o diálogo — por vezes conflituoso — entre essas abordagens.} \cite{maziero2015entendendo}.}
	
Para ilustrar o conceito de entropia sob a ótica da teoria da informação, considere o exemplo no qual você e um amigo estejam lançando um dado honesto. Seu amigo, de costas, tenta adivinhar qual face caiu voltada para cima. Como todas as faces são igualmente prováveis, a chance de acerto é $1/6$. Para ele, a incerteza máxima está presente, e a entropia associada ao sistema é $S_{BG} = k_{B} \ln 6$. Por outro lado, você observa diretamente o resultado e tem certeza do estado do sistema. Nesse caso, a entropia é nula. Esse exemplo ilustra uma característica fundamental da entropia: sua dependência da informação disponível ao observador. Assim, a entropia pode ser interpretada como uma medida da ignorância sobre a configuração de um sistema físico.

	\subsection{PROPRIEDADES DA ENTROPIA DE BOLTZMANN-GIBBS}
	
	Dada uma contextualização histórica e conceitual sobre a entropia de Boltzmann-Gibbs, \textcolor{black}{serão} introduzidas e demonstradas algumas propriedades matemáticas que a fundamentam. Dentre os principais resultados, destaca-se a aditividade, violada pela entropia de Tsallis, como veremos na seção seguinte. Os teoremas de unicidade serão discutidos em detalhes.

	\begin{definition}[Aditividade]
		Seja $\mathcal{F}$ uma grandeza física associada aos sistemas físicos $A$ e $B$, que assumimos serem probabilisticamente independentes. A grandeza $\mathcal{F}(A)$ é associada ao sistema $A$ e $\mathcal{F}(B)$ associada ao sistema $B$. Diz-se que a grandeza $\mathcal{F}$ é \textit{aditiva} se, e somente se,   
		\begin{align*}
			\mathcal{F}(A+B) = \mathcal{F}(A) + \mathcal{F}(B).
		\end{align*}
	\end{definition}
	\begin{ppdd}[Aditividade]
		Considere os subsistemas $A$ e $B$ independentes, com $p_{i}^{(A)}$ e $p_{j}^{(B)}$ representando as probabilidades de cada subsistema acessar o $i$-ésimo e o $j$-ésimo microestado, respectivamente. A probabilidade resultante da composição dos subsistemas é dada por $p_{ij}^{(A+B)} = p_{i}^{(A)} p_{j}^{(B)}$ para todo $(i,j)$, de modo que a entropia de Boltzmann-Gibbs do sistema composto $A+B$ é dada por
		\begin{align}
			&S_{BG}(A+B)=S_{BG}(p_{ij}^{(A+B)}) = \sum_{i,j} p_{ij}^{(A+B)} \ln p_{ij}^{(A+B)} \nonumber\\
			&= \sum_{i,j} p_{i}^{(A)} p_{j}^{(B)} \ln \left(p_{i}^{(A)} p_{j}^{(B)}\right) \nonumber \\
			&= \sum_{i} \sum_{j} p_{i}^{(A)} p_{j}^{(B)} \ln \left(p_{i}^{(A)}\right) + \sum_{i} \sum_{j} p_{i}^{(A)} p_{j}^{(B)} \ln \left( p_{j}^{(B)}\right) \nonumber\\
		&= \sum_{i}p_{i}^{(A)} \ln \left(p_{i}^{(A)}\right) + \sum_{j} p_{j}^{(B)} \ln \left( p_{j}^{(B)}\right) \nonumber \\
			&= S_{BG}(p_{i}^{(A)}) + S_{BG}(p_{j}^{(B)}) = S_{BG}(A) + S_{BG}(B),
		\end{align} 
		e portanto é aditiva. 	
	\end{ppdd}

	\begin{ppdd}[Positividade]
		Se um determinado estado $n_{0}$ do sistema é totalmente conhecido, então tem-se que $p_{n_{0}} = 1$ e $p_{n} = 0$ para todo $n \neq n_{0}$, e, portanto, 
		\begin{align*}
			S_{BG}=k \ln (1)=0.
		\end{align*}
		No entanto, para qualquer outro caso, tem-se que $0<p_{n}<1$, para pelo menos dois valores de $n$ e, consequentemente, $1/p_{n}>1$, resultando, portanto, em
		\begin{align}
			S_{BG} = -k\sum_{n}^{W}p_{n}\ln p_{n} = k\sum_{n}^{W}p_{n}\ln \frac{1}{p_{n}} >0, 
		\end{align} 
		De modo que a entropia de Boltzmann-Gibbs é sempre positiva.
	\end{ppdd}

	\begin{ppdd}[Expansibilidade]\label{ppdd_expansibilidade}
		Se estados possíveis forem adicionados à entropia $S_{BG}$, de modo que a probabilidade, $p_{W+1},\dots,p_{W+N}$, de que eles aconteçam seja nula, então a entropia permanecerá invariante. De fato, temos 
		\begin{align}
			&S_{BG}(p_{1},p_{2},\dots,p_{W},p_{W+1},\dots,p_{W+N}) \\&= -k \sum_{i}^{W+N} p_{i} \ln p_{i} \nonumber\\&= -k\left(\sum_{i}^{W}p_{i}\ln p_{i} + \underbrace{\lim_{p_{i}\to 0}\sum_{W+1}^{N}p_{i}\ln p_{i}}_{\rightarrow 0}\right) \nonumber\\ &= -k\sum_{i}^{W}p_{i}\ln p_{i} = S_{GB}(p_{1},p_{2},\dots,p_{W}).
		\end{align}
		Portanto, temos:
		\begin{align}
			S_{BG}(p_{1},\dots,p_{W},p_{W+1},\dots,p_{W+N})=S_{BG}(p_{1},\dots,p_{W}).
		\end{align}
	\end{ppdd}
O leitor pode, a esta altura, questionar-se sobre a interpretação física da propriedade (\ref{ppdd_expansibilidade}). Para esclarecê-la, consideremos um sistema físico cujos estados possíveis são $\varepsilon_1, \varepsilon_2, \varepsilon_3$ e $\varepsilon_4$. Suponhamos, entretanto, que a energia disponível ao sistema permite o acesso apenas aos estados $\varepsilon_1$ e $\varepsilon_2$. Nesse caso, as probabilidades associadas a esses dois estados são $p_1$ e $p_2$, respectivamente, enquanto os estados $\varepsilon_3$ e $\varepsilon_4$ permanecem inacessíveis, com probabilidades $p_3 = 0$ e $p_4 = 0$. A propriedade de expansibilidade da entropia assegura que esses estados de probabilidade nula podem ser incluídos na soma sem alterar o valor da entropia do sistema.
	
Antes de prosseguirmos para o próximo resultado, será útil revisitar alguns conceitos básicos de análise matemática \cite{lima2004analise}. Seja $f(x): I\subset\mathbb{R}\mapsto\mathbb{R}$ uma função real. Tome os pontos $A=(x_{0},f(x_{0}))$ e $B=(x_{1},f(x_{1}))$ pertencentes ao gráfico $f$, com $x_{0},x_{1}\subset I$. A reta
	\begin{align}
		y(x) = f(x_{0}) + \frac{f(x_{1})-f(x_{0})}{x_{1}-x_{0}}(x-x_{0}), 
	\end{align}
	que liga os pontos $A$ e $B$, é chamada de secante $x_{0}x_{1}$. 
	\begin{figure}[h!]
		\centering
		\includegraphics[width=0.8\linewidth]{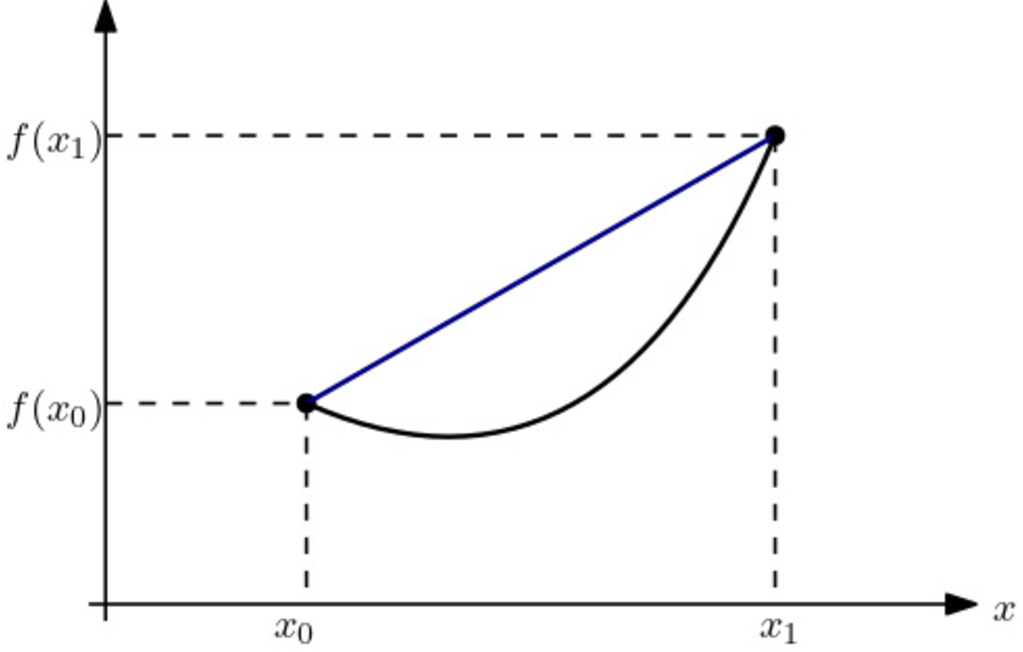}
		\caption{Em azul a reta secante $x_{0}x_{1}$ que liga os pontos $(x_{0},f(x_{0}))$ e $(x_{1},f(x_{1}))$. Em preta a linha que determina o gráfico $f(x)$.}
		\label{fig:funcaoconvexa}
	\end{figure}
	\begin{definition}\label{def_concavidade}
		Uma função $f:I\subset\mathbb{R}\mapsto\mathbb{R}$ é dita ser convexa (côncava) quando seu gráfico está abaixo (acima) de qualquer uma de suas secantes (ver Figura \ref{fig:funcaoconvexa}). Seja $x_{0}<x<x_{1}$ em $I$. Podemos expressar a convexidade (concavidade) de $f$, por meio de \cite{lima2004analise}
		\begin{align}
			f(x)\leq(\geq) f(x_{0}) + \frac{f(x_{1})-f(x_{0})}{x_{1}-x_{0}}(x-x_{0}). 
		\end{align}
	\end{definition}
	
	O ponto $x\in I$ pertencente ao intervalo $[x_{0},x_{1}]$ pode ser escrito de forma única como $x = (1-\lambda)x_{0} + \lambda x_{1}$, com $0 \leq \lambda \leq 1$. Portanto, podemos escrever
	\begin{align}
		\lambda = \frac{(x-x_{0})}{(x_{1}-x_{0})}.
	\end{align}
	Dessa maneira, por meio da definição (\ref{def_concavidade}), podemos afirmar que uma função $f$ é convexa (côncava) se, e somente se,
	\begin{align}
		f(x)\leq(\geq) \lambda f(x_{1}) + (1-\lambda)f(x_{0}). 
	\end{align}
	
Um resultado decorrente das definições acima, que relaciona a derivada de uma função com sua convexidade (concavidade) e que será de suma importância para a propriedade (\ref{BG_fouth_ppdd}), é dado por \cite{lima2004analise}:
	
	\begin{corollary}\label{corolario_segunda_derivada}
		Uma função $f:I\subset\mathbb{R}\mapsto\mathbb{R}$, diferenciável duas vezes no intervalo $I$, é convexa (côncava) se, e somente se, $f''(x)\geq (\leq)  0$ para todo $x \in I$.
	\end{corollary}
	
	Com isso, podemos discutir a próxima propriedade associada à entropia de Boltzmann-Gibbs.

	\begin{ppdd}[Concavidade]\label{BG_fouth_ppdd}
		Considere dois conjuntos de probabilidades $\{p_{i}\}$ e $\{\tilde{p}_{i}\}$ associados ao mesmo sistema, com $W$ microestados. Pode-se definir um novo conjunto de probabilidades intermediárias $\{p'_{i}\}$, dado da seguinte maneira
		\begin{align}\label{prob_intermediaria}
			p_{i}' = \lambda p_{i} + (1-\lambda)\tilde{p}_{i} \qquad (\forall i; 0<\lambda<1).
		\end{align} 
		A entropia $S_{BG}$ é dita ser côncava se, e somente se,
		\begin{align}\label{def_concavidade_bg}
			S_{BG}(\{p_{i}'\}) > \lambda S_{BG}(\{p_{i}\}) + (1-\lambda)S_{BG}(\{\tilde{p}_{i}\}).
		\end{align}
		Como a função $-x \ln x$ possui derivada segunda negativa, ela satisfaz a equação \eqref{def_concavidade_bg}, e, portanto, tem-se que
		\begin{align}\label{eq_caqui}
			-p_{i}' \ln p_{i}' > \lambda (-p_{i} \ln p_{i}) + (1-\lambda)(-\tilde{p}_{i}\ln \tilde{p}_{i}) \nonumber\\ (\forall i; 0<\lambda<1).
		\end{align}
		Realizando a soma sobre todos os estados $\sum_{i=1}^{W}$ na equação (\ref{eq_caqui}), a equação \eqref{def_concavidade_bg} é imediatamente recuperada, concluindo-se que a entropia de Boltzmann-Gibbs é côncava.
	\end{ppdd}

Após a apresentação e demonstração de algumas propriedades fundamentais da entropia de Boltzmann-Gibbs, passamos agora à análise de resultados de unicidade que essa entropia deve satisfazer. Nos trabalhos de Shannon \cite{book:320619}, em 1948, e de Khinchin \cite{book:111623}, em 1953, foi proposto um conjunto de axiomas que estabelecem critérios para a forma funcional admissível da entropia. Sob hipóteses matematicamente razoáveis, demonstra-se que a única função que satisfaz integralmente esses axiomas é justamente a entropia de Boltzmann-Gibbs.

\subsubsection*{TEOREMA DE SHANNON}

{Shannon propôs o famoso teorema que carrega seu nome no contexto da teoria da informação, sem nenhuma conexão direta com a termodinâmica ou com a mecânica estatística de Boltzmann-Gibbs. Ele tinha em mente definir uma quantidade que determinasse a taxa de informação produzida por um sistema. Essa proposta dialoga com a noção de entropia discutida na seção \ref{def_entropia_de_bolt_gibbs_shannon}, que a interpreta como uma medida da ignorância de um determinado observador sobre o sistema. Suponha que $p_{1},p_{2},\dots,p_{W}$ sejam as probabilidades associadas aos eventos possíveis de um experimento aleatório.} {Se existir uma medida $S(p_{1},p_{2},\dots,p_{W})$, que informe o quanto de escolha está envolvida na seleção dos eventos ou a quantidade de incerteza sobre o resultado desse experimento, é razoável requerer certas propriedades sobre essa quantidade \cite{book:320619}:
		\begin{enumerate}
			\item \label{primeira_cond_shannon} $S(p_{1},p_{2},\dots,p_{W})$ é uma função contínua de $\{p_{i}\}$;
			\item \label{segunda_cond_shannon} Para o conjunto $\{p_{i}\}$ de probabilidades equiprováveis, com $p_{i} = \frac{1}{W}$, a função $S(p_{i}=1/W)$, $\forall i$, deve crescer monotonamente e ser contínua com respeito a $W$;
			\item \label{terceira_cond_shannon} Se uma escolha for dividida em duas escolhas sucessivas, a função $S({p_{i}})$ original deve ser a soma ponderada dos valores individuais de $S$.
	\end{enumerate}} 
	
Para compreender de forma mais intuitiva a condição~(\ref{terceira_cond_shannon}), considere um experimento com três possíveis resultados, cada um com uma probabilidade distinta de ocorrência: $p_{1} = \frac{1}{2}$, $p_{2} = \frac{1}{3}$ e $p_{3} = \frac{1}{6}$, conforme ilustrado na Figura~\ref{fig:diagramaexperimentos}. Nesse cenário, a entropia associada ao sistema é dada por $S\left(\frac{1}{2}, \frac{1}{3}, \frac{1}{6}\right)$.
	\begin{figure}[h!]
		\centering
		\includegraphics[width=0.9\linewidth]{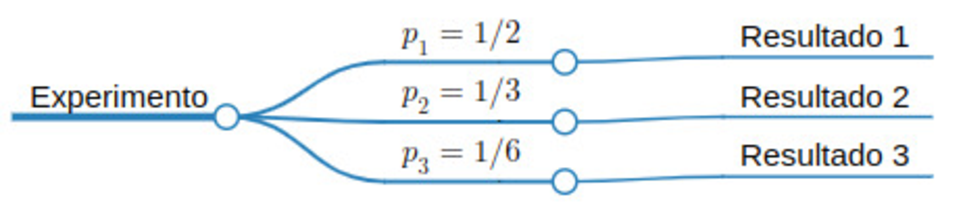}
		\caption{Diagrama ilustrando o experimento com as três possíveis resultados e as respectivas probabilidades associadas: $p_{1}$, $p_{2}$ e $p_{3}$.}
		\label{fig:diagramaexperimentos}
	\end{figure}
Considere agora um segundo experimento, no qual a primeira etapa envolve duas opções igualmente prováveis, cada uma com probabilidade $p'_{1} = p'_{2} = \frac{1}{2}$. Caso a primeira opção seja escolhida $\left(p'_{1} = \frac{1}{2}\right)$, o experimento leva diretamente ao resultado 1 — o mesmo do primeiro experimento ilustrado na Figura~\ref{fig:diagramaexperimentos}. Já a segunda opção $\left(p'_{2}=\frac{1}{2}\right)$ conduz a uma segunda etapa, na qual dois novos resultados podem ocorrer, com probabilidades $p_{21} = \frac{2}{3}$ e $p_{22} = \frac{1}{3}$, correspondendo aos resultados 2 e 3, respectivamente. Esse segundo experimento pode ser representado como mostrado na Figura~\ref{fig:diagramaexperimentos2}.
	\begin{figure}[h!]
		\centering
		\includegraphics[width=1\linewidth]{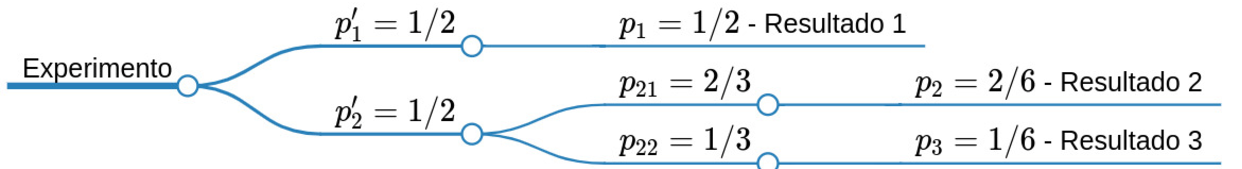}
		\caption{Diagrama ilustrando o experimento. Inicialmente, há duas possibilidades de escolhas. A escolha com probabilidade $p'_{1}=1/2$ leva ao resultado 1. Já a escolha $p'_{2}=1/2$ se subdivide em duas outras escolhas possíveis, com probabilidades $p_{21}=1/3$ e $p_{22}=2/3$, cada uma levando aos resultados 2 e 3, respectivamente.}
		\label{fig:diagramaexperimentos2}
	\end{figure}
	{A entropia associada ao sistema deste segundo experimento, à luz da condição (\ref{terceira_cond_shannon}), será dada por 
		\begin{align}
			S\left(\frac{1}{2},\frac{1}{3},\frac{1}{6}\right) = S\left(\frac{1}{2},\frac{1}{2}\right) + \frac{1}{2}S\left(\frac{2}{3},\frac{1}{3}\right).
	\end{align}}
		
		{\begin{theorem}[Shannon]
				A única função $S(\{p_{i}\})$ que satisfaz as propriedades (\ref{primeira_cond_shannon}), (\ref{segunda_cond_shannon}) e (\ref{terceira_cond_shannon}), listadas acima, é da forma 
				\begin{align}
					S(\{p_{i}\}) = -k\sum_{i=1}^{W} p_{i}\ln p_{i}. 
				\end{align}
		\end{theorem}}
		\begin{proof}
			A demonstração aqui exposta é baseada em \cite{carter2001classical}. 
			Considere que a função $S(p_{1},p_{2},\dots,p_{W})$ possa ser escrita como 
			\begin{align}\label{ceu}
				S(p_{1},p_{2},\dots,p_{W}) = \sum_{i=1}^{W} f(p_{i}), 
			\end{align} 
			onde $f(p_{i})$ é uma função contínua a ser determinada. Como ela é válida para todos os casos, inclusive para os equiprováveis, é suficiente determiná-la nesse caso, em que $p_{i}=1/W$, e 
			\begin{align}\label{maca}
				S\left(\frac{1}{W},\frac{1}{W},...,\frac{1}{W}\right) = W f\left(\frac{1}{W}\right).
			\end{align}
			Pela condição (\ref{segunda_cond_shannon}), tem-se que 
			\begin{align}\label{caqui}
				\frac{d}{dW}\left[W f(\frac{1}{W})\right] \geq 0.
			\end{align}
			Utilizando a condição (\ref{terceira_cond_shannon}) e considerando um experimento com $m$ possibilidades iguais, em que cada escolha possa ser decomposta em $n$ escolhas de possibilidades iguais, com $W=nm$, tem-se
			\begin{align}\label{pera}
				&S\left(\frac{1}{W},\frac{1}{W},...,\frac{1}{W}\right) = S\left(\frac{1}{m},\frac{1}{m},...,\frac{1}{m}\right) + S\left(\frac{1}{n},\frac{1}{n},...,\frac{1}{n}\right) \nonumber\\&= S\left(\frac{1}{nm},\frac{1}{nm},...,\frac{1}{nm}\right).
			\end{align}
			Escrevendo a equação (\ref{pera}) por meio da equação (\ref{maca}), tem-se
			\begin{align}
				m f\left(\frac{1}{m}\right) + nf\left(\frac{1}{n}\right) = mnf\left(\frac{1}{mn}\right).
			\end{align}
			Adotando $M = \frac{1}{m}$ e $N = \frac{1}{n}$, pode-se reescrever a equação \ref{pera} como
			\begin{align}
				\frac{1}{M}f(M) + \frac{1}{N}f(N) = \frac{1}{MN}f(MN).
			\end{align}
			Definindo $h(N) = \frac{1}{N}f(N)$, $h(M) = \frac{1}{M}f(M)$ e $h(MN) = \frac{1}{MN}f(MN)$, obtém-se 
			\begin{align}\label{uva}
				h(N) + h(M) = h(MN).
			\end{align}
			Diferenciando a equação (\ref{uva}) em relação a $M$ e também em relação a $N$, tem-se
			\begin{align}
				h'(M) =  \frac{\partial h(MN)}{\partial M}=\frac{\partial(MN)}{\partial M}\frac{\partial h(MN)}{\partial(MN)}=N\frac{\partial h(MN)}{\partial(MN)},
			\end{align} 
			\begin{align}
				h'(N) = \frac{\partial h(MN)}{\partial N}=\frac{\partial(MN)}{\partial N}\frac{\partial h(MN)}{\partial(MN)}=M\frac{\partial h(MN)}{\partial(MN)}.
			\end{align}
			Dessas relações obtém-se
			\begin{align}
				M h'(M) = N h'(N).
			\end{align}
			Como as variáveis $M$ e $N$ são independentes, tem-se
			\begin{align}
				M h'(M) = A,
			\end{align}
			onde $A=constante$. E, portanto,
			\begin{align}
				h(M) = A\ln M + C,
			\end{align}
			com $C=constante$. Restaurando as variáveis originais 
			\begin{align}
				m f\left(\frac{1}{m}\right) = A \ln \left(\frac{1}{m}\right) + C,
			\end{align}
			\begin{align}
				f\left(\frac{1}{m}\right) = \frac{A}{m} \ln \left(\frac{1}{m}\right) + \frac{C}{m}.
			\end{align}
			Se a probabilidade de obter um determinado resultado em um experimento for $1$, a incerteza sobre esse resultado deve ser nula. Logo, $f(1)=0$, o que conduz à
			\begin{align}
				f(1) = C = 0,
			\end{align}
			portanto,
			\begin{align}
				f\left(\frac{1}{m}\right) = -\frac{A}{m} \ln \left({m}\right).
			\end{align}
			Usando a condição (\ref{segunda_cond_shannon}) dada pela equação (\ref{caqui}), tem-se
			\begin{align}
				\frac{d}{dm} \left[m f\left(\frac{1}{m}\right)\right] = -\frac{A}{m} \geq 0,
			\end{align}
			então, $A$ deve ser negativa. Fazendo a suposição de que $A=-k$, com $k>0$, e restaurando $p = 1/m$
			\begin{align}
				f(p) = -k p\ln p,
			\end{align}
			Usando a equação (\ref{ceu}), encontra-se 
			\begin{align}
				S(p_{i}) = -k \sum_{i=1}^{W} p_{i} \ln p_{i}.
			\end{align}
			Isso finaliza a demonstração.
		\end{proof}

		\subsubsection*{TEOREMA DE KHINCHIN}

		{Considere o caso em que dois sistemas, $A$ e $B$, são probabilisticamente dependentes, ou seja, o que ocorre com $A$ depende do que ocorre com $B$, ou vice-versa. Será denotada por $q_{k\ell}$ a probabilidade de um evento $B_{\ell}$ acontecer em $B$, dado que o evento $A_{k}$, com probabilidade $p_{k}$, já ocorreu em $A$. Para exemplificar, considere o lançamento de dois dados. Qual é a probabilidade de, após o lançamento, a soma das faces ser 8, dado que as faces são dois números ímpares? Nesse exemplo, a probabilidade da soma ser 8 está condicionada ao fato de as faces dos dados serem ímpares. Logo, o evento $A$ será denotado pela soma ser 8, e o evento $B$, condicionado ao evento $A$, é a probabilidade de as faces serem ímpares. A probabilidade conjunta dos sistemas $A$ e $B$ é dada por
			\begin{align}
				\pi_{k\ell} = p_{k}q_{k\ell}.
		\end{align}}
		
		{Utilizando o sistema conjunto $A$ e $B$, pode-se calcular a entropia a partir da probabilidade conjunta de $A$ e $B$ 
			\begin{align}
				S(A+B) &= - \sum_{k}\sum_{\ell} p_{k}q_{k\ell} \ln \left(p_{k}q_{k\ell}\right) \\
				&=- \sum_{k}\sum_{\ell} p_{k}q_{k\ell}\left[\ln p_{k} + \ln q_{k\ell}\right] .
			\end{align}
			 Usando a condição de normalização $\sum_{\ell} q_{k\ell}=1$ e definindo $S_{k}\equiv - \sum_{\ell} q_{k\ell}\ln q_{k\ell}$, obtém-se
			\begin{align}
				S(A+B) = S(A) + \sum_{k}p_{k}S_{k}(B).
			\end{align} 
			Definindo $S(B|A)\equiv \sum_{k}p_{k}S_{k}(B)$, como a entropia condicional do sistema $B$ condicionado ao sistema $A$, reescreve-se
			\begin{align}
				S(A + B) = S(A) + S(B|A).
			\end{align}
			No caso em que $A$ e $B$ são sistemas probabilisticamente independentes, a equação acima se reduz a
			\begin{align}
				S(A + B) = S(A) + S(B).
		\end{align}}
		
		{Para estabelecer o teorema de unicidade de Khinchin, considere a existência de um funcional $S(p_{i})$ que satisfaça as seguintes propriedades \cite{book:111623}:
			\begin{enumerate}
				\item \label{primeira_cond_khinchin} Para um dado $W$ e para $\sum_{i=1}^{W} p_{i}$, a função $S(p_{1},p_{2},\dots,p_{W})$ toma seu maior valor para $p_{i}=\frac{1}{W}$ $(i=1,2,\dots,W)$; 
				\item \label{segunda_cond_khinchin} S(A + B) = S(A) + S(B|A); 
				\item \label{terceira_cond_khinchin} $S(p_{1},p_{2},\dots,p_{W},0)=S(p_{1},p_{2},\dots,p_{W})$;
			\end{enumerate}
			\begin{theorem}[Khinchin]
				Seja $S(p_{1},p_{2},\dots,p_{W})$ uma função definida para qualquer inteiro $W$ e para todos os valores $p_{1},p_{2},\dots,p_{W}$ tais que $p_{i}\geq 0\quad (i=1,2,\dots,W)$, $\sum_{i=1}^{W} p_{i}=1$. Se para qualquer $W$ essa função for contínua com respeito a todos os seus argumentos e se ela satisfizer as propriedades (\ref{primeira_cond_khinchin}), (\ref{segunda_cond_khinchin}) e (\ref{terceira_cond_khinchin}), então 
				\begin{align}
					S(p_{1},p_{2},\dots,p_{W}) = -k \sum_{i=1}^{W} p_{i} \ln p_{i} \qquad (k>0).
				\end{align}
			\end{theorem}
			A demonstração do teorema enunciado acima pode ser encontrada em \cite{book:111623}.}
		
		\subsection{ENSEMBLE MICROCANÔNICO E ENSEMBLE CANÔNICO}
		
A entropia de Boltzmann-Gibbs, expressa pela equação~\eqref{entropia_BG_estados_discretos}, pode ser associada a diferentes tipos de sistemas físicos. Em física estatística, dois ensembles introdutórios são frequentemente analisados: o ensemble microcanônico e o ensemble canônico. O ensemble microcanônico representa um sistema idealizado, completamente isolado do ambiente, no qual a energia interna, o número de partículas e o volume permanecem constantes. Por outro lado, no ensemble canônico, o sistema pode trocar energia com o meio externo. É, portanto, considerado em contato térmico com um reservatório a temperatura constante. Neste caso, a energia interna média do sistema é dada por
\begin{align}\label{energia_media_ensemble_canonico}
	\langle E\rangle = \sum_{i=1}^{W} p_{i}\varepsilon_{i},
\end{align}
onde $p_i$ representa a probabilidade de o sistema estar no estado $\varepsilon_i$. Quando o volume do sistema $V$ é muito menor que o volume do reservatório $V_R$, o comportamento energético do sistema pode ser descrito de forma confiável por esse formalismo.

Do ponto de vista da teoria da informação \cite{Jaynes:1957zza, Cover:2005lom,natal2021entropy}, o ensemble microcanônico corresponde a um cenário de incerteza máxima, no qual todos os microestados acessíveis são equiprováveis e nenhuma informação adicional está disponível, levando a uma entropia máxima $S = \ln W$. Já o ensemble canônico se associa à situação em que se deseja maximizar a entropia sob uma restrição de valor médio (como a energia), resultando em uma distribuição de probabilidades que expressa conhecimento parcial sobre o sistema. 

A análise desses ensembles permite explorar propriedades matemáticas e físicas relevantes da entropia, conforme o tipo de interação do sistema com o ambiente.
		
		\begin{proposition}[Ensemble microcanônico]
			Em um sistema físico isolado, a entropia de Boltzmann-Gibbs \eqref{entropia_BG_estados_discretos}, sob a condição de normalização da probabilidade \eqref{condicao_de_normalizacao}, assume seu valor máximo dado pela equação \eqref{entropia_BG_estados_igual}. A probabilidade é dada por $p_{i}=1/W$, onde $W$ é o número de estados acessíveis ao sistema.
		\end{proposition} 
		
		\begin{proof}
			Para demonstrarmos essa afirmação, utilizaremos o método dos multiplicadores de Lagrange. Considere uma função dada por
			\begin{align}
				\mathcal{L} &= S_{BG} -\lambda\left(1 - \sum_{i=1}^{W}p_{i}\right) \nonumber\\&= -k\sum_{i=1}^{W}p_{i}\ln p_{i}-\lambda\left(1 - \sum_{i=1}^{W}p_{i}\right).
			\end{align}
			A variação de $\mathcal{L}$ nos fornecerá 
			\begin{align}
				\delta\mathcal{L} &= \sum_{i=1}^{W}\left(-k\delta p_{i}\ln p_{i}  - k\delta p_{i} + \lambda\delta p_{i}\right) \\&=\sum_{i=1}^{W}(-k\ln p_{i} - k + \lambda)\delta p_{i}= 0,
			\end{align}
			o que resulta
			\begin{align}
				-k\ln p_{i} - k + \lambda = 0, 
			\end{align}
			logo,
			\begin{align}\label{ppdd_normalizacao_exp}
				p_{i} = \exp\left(\frac{\lambda - k}{k}\right).
			\end{align}
			Usando a condição de normalização
			\begin{align}
				\sum_{i=1}^{W} p_{i} = \sum_{i=1}^{W} \exp\left(\frac{\lambda - k}{k}\right) = W\exp\left(\frac{\lambda - k}{k}\right)= 1,
			\end{align}
			obtemos
			\begin{align}
				\exp\left(\frac{\lambda -k}{k}\right) = \frac{1}{W},
			\end{align}
			\begin{align}\label{exp_para_lambda_microcanonico}
				\lambda = k\ln\left(\frac{1}{W}\right) + k.
			\end{align}
			Substituindo a equação \eqref{exp_para_lambda_microcanonico} na equação \eqref{ppdd_normalizacao_exp}, obtém-se
			\begin{align}
				p_{i} = \frac{1}{W}.
			\end{align}
			Logo a entropia de Boltzmann-Gibbs assume a forma
			\begin{align}\label{entropy_microcanonical_ensemble}
				S_{BG} = -k\sum_{i=1}^{W}\frac{1}{W}\ln\left(\frac{1}{W}\right) = k\ln W . 
			\end{align}
			
			Derivando duas vezes a equação acima, obtemos
			\begin{align}
				\frac{\partial^{2} S_{BG}}{\partial W^{2}} = -\frac{k}{W^{2}}<0,
			\end{align}
			e, portanto, concluí-se que a entropia assume seu valor máximo.
		\end{proof}
		
		\begin{proposition}[Ensemble canônico]
			Um sistema físico em contato com um reservatório térmico, que satisfaz a entropia de Boltzmann-Gibbs \eqref{entropia_BG_estados_discretos}, sob a condição de normalização \eqref{condicao_de_normalizacao}, com energia média dada pela equação \eqref{energia_media_ensemble_canonico} e à luz do princípio da máxima entropia, é caracterizado pela função de partição dada por
			\begin{align}
				Z = \sum_{i=1}^{W} e^{-\frac{\beta \varepsilon}{k}}.
			\end{align}
		\end{proposition}
		
		\begin{proof}
			Para realizar a demonstração da proposição acima, utilizaremos novamente o método de multiplicadores dos Lagrange. Para isso, introduziremos a função 
			\begin{align}
				\mathcal{L} &= S_{BG} - \alpha \left(1 - \sum_{i=1}^{W}p_{i}\right) -\beta \left(\langle E \rangle - \sum_{i=1}^{W}\varepsilon_{i}p_{i}\right) \\&= -k\sum_{i=1}^{W} p_{i}\ln p_{i} - \alpha \left(1 - \sum_{i=1}^{W}p_{i}\right) -\beta \left(\langle E \rangle - \sum_{i=1}^{W}\varepsilon_{i}p_{i}\right).
			\end{align}
			A variação da função acima fornece 
			\begin{align}
				\delta \mathcal{L}  &= \sum_{i=1}^{W}\left(-k\ln p_{i}\delta{p_{i}} -k\delta p_{i} +\alpha\delta p_{i} +\beta\varepsilon_{i}\delta p_{i}\right) \\ &= \sum_{i=1}^{W}\left(-k\ln p_{i} -k +\alpha +\beta\varepsilon_{i}\right)\delta p_{i} = 0,
			\end{align}
			resultando em
			\begin{align}\label{olho_olho}
				\left(-k\ln p_{i} -k +\alpha +\beta\varepsilon_{i}\right) = 0.
			\end{align}
			Da equação (\ref{olho_olho}), obtemos que
			\begin{align}
				p_{i} = \exp\left(\frac{k - \alpha - \beta\varepsilon_{i}}{k}\right).
			\end{align}
			Usando a condição de normalização da probabilidade, temos 
			\begin{align}
				\sum_{i=1}^{W} p_{i} =  \sum_{i=1}^{W} \exp\left(\frac{k - \alpha - \beta\varepsilon_{i}}{k}\right) = 1,
			\end{align}
			\begin{align}
				\exp\left(\frac{k - \alpha }{k}\right) = \frac{1}{\sum_{i={1}}^{W}\exp\left(\frac{- \beta\varepsilon_{i}}{k}\right)}.
			\end{align}
			Logo, podemos escrever a probabilidade como 
			\begin{align}\label{ppbb_fun_part_BG}
				p_{i} = \frac{e^{-\frac{\beta\varepsilon_{i}}{k}}}{Z},
			\end{align}
			onde 
			\begin{align}
				Z = \sum_{i=1}^{W}e^{-\frac{\beta\varepsilon_{i}}{k}},
			\end{align}
			é chamada de função de partição. Substituindo a equação \eqref{ppbb_fun_part_BG} para a probabilidade na equação da entropia de Boltzmann-Gibbs, obtemos
			\begin{align}
				S_{BG} &= -k \sum_{i=1}^{W} p_{i} \ln\left(\frac{e^{-\frac{\beta\varepsilon_{i}}{k}}}{Z}\right) \\&= -k \sum_{i=1}^{W} p_{i} \left[\ln\left(e^{-\frac{\beta\varepsilon_{i}}{k}}\right) - \ln \left(Z\right)\right] \\ &= -k \sum_{i=1}^{W} p_{i} \left[-\frac{\beta\varepsilon_{i}}{k} - \ln \left(Z\right)\right] \\&= \beta\sum_{i=1}^{W} p_{i}\varepsilon_{i} + k\ln\left(Z\right).
			\end{align}
			Usando a equação \eqref{energia_media_ensemble_canonico} para a energia, encontramos que a entropia é dada por
			\begin{align}\label{entropy_canonical_ensemble}
				S_{BG} = \beta \langle E \rangle + k\ln\left(Z\right).
			\end{align}
		\end{proof}

As equações \eqref{entropy_microcanonical_ensemble} e \eqref{entropy_canonical_ensemble}, encontradas para a entropia nos ensembles microcanônico e canônico, respectivamente, não apresentam relações explícitas e não se restringem à termodinâmica. Essa conexão entre os ensembles e a termodinâmica pode ser encontrada de forma detalhada no capítulo 17 da referência \cite{book:18204} e na referência \cite{book:1322644}. Ou nas referências \cite{Jaynes:1957zza, Cover:2005lom} para o contexto da teoria da informação.

\section{MECÂNICA ESTATÍSTICA NÃO ADITIVA}\label{mec_ext_nao_ext}

A mecânica estatística de Boltzmann-Gibbs, apesar de central na física, está calcada no fato de que a entropia é aditiva para a junção de dois sistemas probabilisticamente independentes, além de não ser universal. Por outro lado, de acordo com Tsallis, ``\textit{a entropia é um conceito delicado e poderoso, construído cuidadosamente para uma classe de sistemas}'' \cite{book:166140}. Nesse sentido, existe uma classe de sistemas físicos para os quais a entropia de Boltzmann-Gibbs pode ser adequada no contexto da termodinâmica, como o gás de Van der Waals. Mas para outros, como o átomo de hidrogênio \cite{PhysRevE.51.6247}, ela não é.

Abordagens alternativas para o conceito de entropia têm sido amplamente investigadas na literatura. Dentre elas, destaca-se a proposta introduzida por Tsallis em seu influente trabalho de 1988~\cite{Tsallis:1987eu}. Nessa formulação, a entropia de Boltzmann-Gibbs é generalizada por meio de um parâmetro real $q$, que caracteriza o grau de não extensividade do sistema. Essa generalização viola a aditividade usual da entropia e permite acomodar propriedades específicas de certos sistemas complexos, como aquelas relacionadas à dinâmica microscópica~\cite{tsallis2009nonadditive}. Outras formulações generalizadas também têm sido exploradas, incluindo a entropia de Rényi~\cite{renyi1970probability} e a entropia de Tsallis-Cirto~\cite{Tsallis:2012js}. Esta última, proposta como uma alternativa para restaurar a extensividade da entropia de Bekenstein-Hawking \cite{Bekenstein:1972tm, Hawking:1975vcx}, no contexto da física de buracos negros.

A entropia de Boltzmann-Gibbs não é derivada diretamente da dinâmica microscópica do sistema, sendo, portanto, assumida como um postulado  da teoria estatística. No entanto, uma vez estabelecida a dinâmica microscópica, todas as demais grandezas estatísticas podem ser determinadas a partir dela. Dado isso, qualquer proposta de generalização da entropia também não deve partir de uma descrição microscópica, mas sim ser introduzida de maneira axiomática ou heurística. Nesse contexto, uma abordagem possível consiste no uso de metáforas ou princípios inspiradores, como exemplificado na proposta de Tsallis~\cite{tsallis2003introduction,Tsallis:2019giw}.

Considere a equação diferencial mais simples possível. Concluímos, sem muita elaboração 
\begin{align}\label{eq_dif_0}
	y'(x) = 0,
\end{align}
com $y'(x)=\frac{dy}{dx}$ e $y=\text{constante}$. Outra equação simples que pode ser considerada é 
\begin{align}\label{eq_dif_1}
	y'(x) = 1,
\end{align}
cuja solução é dada por $y = x + \text{constante}$, onde a condição inicial sempre pode ser ajustada para que a constante seja nula $(y(0)=0)$. Na tentativa de tornar menos restritas as equações acima, pode-se propor a seguinte equação diferencial 
\begin{align}\label{eq_dif_y}
	y'(x) = y,
\end{align}
onde a solução, com a condição inicial $x(1)=0$, é dada por 
\begin{align}\label{solucao_dy_dx_y}
	x=\ln y,
\end{align}
e a inversa é $y =e^{x}$. A solução obtida na equação~\eqref{solucao_dy_dx_y} apresenta a mesma forma funcional da entropia de Boltzmann-Gibbs, preservando, inclusive, a propriedade de aditividade.
\begin{align}
	\ln(x_{A}x_{B}) = \ln x_{A} + \ln x_{B}.
\end{align}
Uma generalização das equações diferenciais discutidas acima pode ser dada por
\begin{align}
	y'(x) = y^q \qquad (q\in\mathbb{R}),
\end{align}
onde a solução é dada pelo \textit{q-logaritmo}
\begin{align}
	x = \frac{y^{1-q} - 1}{1-q} \equiv \ln_{q}y,
\end{align}
e a inversa é dada pela $q-exponencial$, definida por 
\begin{align}
	y = [1 + (1-q)x]^{\frac{1}{1-q}}\equiv e_{q}^x.
\end{align}
Os resultados das equações diferenciais \eqref{eq_dif_0}, \eqref{eq_dif_1} e \eqref{eq_dif_y} podem ser retomados para $q\to -\infty$, $q=0$ e $q=1$, respectivamente. A solução satisfaz a chamada pseudo-aditividade
\begin{align}\label{pseudo_aditividade_q_logaritmo}
	\ln_{q} (x_{A}x_{B}) = \ln_{q} x_{A} + \ln_{q} x_{B} + (1-q)\ln_{q} x_{A}\ln_{q} x_{B}.
\end{align}
A abordagem apresentada fornece um caminho natural para a generalização da entropia de Boltzmann-Gibbs, levando à chamada \textit{q-entropia}, também conhecida como entropia de Tsallis. Essa generalização é construída a partir do uso do \textit{q-logaritmo}, de modo que a entropia de Tsallis é definida como
\begin{align}\label{Tsallis_entropia_estados_discretos}
	S_{q}(p_{i}) &= k\frac{1 - \sum_{i=1}^{W} p_{i}^{q}}{q-1}=k\sum_{i=1}^{W} p_{i}\frac{(1-p_{i}^{q-1})}{q-1}\nonumber\\&=k\sum_{i=1}^{W} p_{i}\frac{[1-\left(\frac{1}{p_{i}}\right)^{1-q}]}{q-1}= k\sum_{i=1}^{W} p_{i}\ln_{q} \left(\frac{1}{p_{i}}\right),  
\end{align}
com a condição de normalização
\begin{align}
	\sum_{i}^{W}p_{i}=1.
\end{align}

{Reescrevendo a definição da entropia de Tsallis, dada pela equação \eqref{Tsallis_entropia_estados_discretos}, obtemos
	\begin{align}
		S_{q}(\{p_{i}\}) &= k\frac{1 - \sum_{i}p_{i}p_{i}^{q-1}}{q-1}\\
		&= k\frac{1 - \sum_{i}p_{i}\exp[(q-1)\ln p_{i}]}{q-1}\label{eq_entrop_tsallis_p_limite},
	\end{align}	
	sendo $k$ uma constante arbitrária. Pode-se expandir a exponencial da equação \eqref{eq_entrop_tsallis_p_limite} em série de Taylor em torno de $q=1$
	\begin{align}
		S_{q}(\{p_{i}\}) &\approx k\frac{1 - \sum_{i}p_{i}[1 + (q-1)\ln p_{i}]}{q-1}\\
		&= k\frac{1 - \sum_{i}p_{i} - (q-1)\sum_{i} p_{i}\ln p_{i}}{q-1} \\
		&= -k\sum_{i} p_{i} \ln p_{i} = S_{BG}(\{p_{i}\}).
	\end{align}
	Mostra-se, então, que a entropia de Boltzmann-Gibbs é retomada em uma vizinhança de $q\to 1$. De modo que a entropia de Tsallis é de fato uma generalização, que contém a entropia de Boltzmann-Gibbs como um caso particular.}

Para estados equiprováveis, onde $p_{i}=1/W$, a entropia de Tsallis torna-se
\begin{align}\label{entropia_tsallis_equiprobabilistica}
	S_{q}(W) = k\frac{W^{1-q}-1}{1-q} = k\ln_{q} W,
\end{align}
analogamente, como dada pela equação para estados equiprováveis de Boltzmann-Gibbs \eqref{entropia_BG_estados_igual}.

Como a entropia de Tsallis recupera a entropia de Boltzmann-Gibbs em um limite específico, sua interpretação permanece análoga à discutida na subseção~\ref{def_entropia_de_bolt_gibbs_shannon}. No entanto, nesse contexto, os sistemas considerados apresentam correlações características entre seus componentes, o que justifica a necessidade de uma formulação mais geral. O parâmetro $q$ está associado à essa correlação. A entropia de Tsallis pode, assim, ser interpretada como uma medida da ignorância em relação ao estado de um sistema físico, estendendo a aplicabilidade da abordagem mecânica estatística clássica. A motivação para essa generalização é justamente acomodar cenários físicos que escapam à descrição tradicional de Boltzmann-Gibbs. \textcolor{black}{Nas seções seguintes estabeleceremos os aspectos formais da entropia de Tsallis. Algumas de suas aplicações podem ser encontradas na Seção \ref{aplicacoes}. Para uma discussão mais ampla sobre aplicações, recomendamos as referências ~\cite{tsallis_2021} e ~\cite{book:166140}.
}

\subsection{PROPRIEDADES DA ENTROPIA DE TSALLIS}

Dado que a entropia de Tsallis é uma generalização baseada no \textit{q-logaritmo}, é natural que não seja aditiva, assim como expresso na equação \eqref{pseudo_aditividade_q_logaritmo}. De fato, ela satisfaz a pseudo-aditividade.
\begin{ppdd}[Pseudo-aditividade]
	Dados dois sistemas probabilisticamente independentes, $A$ e $B$, tal que a probabilidade conjunta desses dois sistemas seja dada por $p_{ij}^{(A+B)}=p_{i}^{(A)}p_{j}^{(B)}$, a soma da entropia desses sistemas compostos é fornecida por
	{\begin{align}\label{pbbs_ij_tsallis}
			\sum_{i,j} \left[p_{ij}^{(A+B)}\right]^{q} = \sum_{i} \left[p_{i}^{(A)}\right]^{q}\sum_{j} \left[p_{j}^{(B)}\right]^{q},
		\end{align}
		tomando o logaritmo da equação \eqref{pbbs_ij_tsallis}, obtém-se
		\begin{align}
			&\ln\left\{\sum_{i,j} \left[p_{ij}^{(A+B)}\right]^{q}\right\} \nonumber\\&= \ln\left\{\sum_{i} \left[p_{i}^{(A)}\right]^{q}\right\} + \ln\left\{\sum_{j} \left[p_{j}^{(B)}\right]^{q}\right\}.
		\end{align}
		Escrevendo cada termo da equação acima em termos da entropia de Tsallis \eqref{Tsallis_entropia_estados_discretos}, tem-se
		\begin{align}\label{logaritmo_formal_tsallis}
			&\ln\left[1 + (1-q)\frac{S_{q}(A+B)}{k}\right] \nonumber\\&= \ln\left[1 + (1-q)\frac{S_{q}(A)}{k}\right] + \ln\left[1 + (1-q)\frac{S_{q}(B)}{k}\right],
		\end{align}
		onde $S_{q}(A)=S_{q}(\{p_{i}^{(A)}\})$, $S_{q}(B)=S_{q}(\{p_{i}^{(B)}\})$ e $S_{q}(A+B)=S_{q}(\{p_{ij}^{(A+B)}\})$. Reescrevendo a equação (\ref{logaritmo_formal_tsallis}) como
		\begin{align}
			&1 + (1-q)\frac{S_{q}(A+B)}{k} \nonumber\\&= \left[1 + (1-q)\frac{S_{q}(A)}{k}\right]\left[1 + (1-q)\frac{S_{q}(B)}{k}\right],
		\end{align}
		ela se reduzirá a
	}
	\begin{align}\label{def_pseudo_aditividade}
		\frac{S_{q}(A+B)}{k} = \frac{S_{q}(A)}{k} + \frac{S_{q}(B)}{k} + (1-q)\frac{S_{q}(A)}{k}\frac{S_{q}(B)}{k}.
	\end{align}
	Têm-se os seguintes casos correspondentes para os valores do parâmetro $q$ \cite{book:166140}:

\begin{table}[h!]
	\centering
	\vspace{0.5em}
	\begin{tabular}{l|c|c}
		\toprule
		\textbf{} & \textbf{Parâmetro $q$} & \textbf{Relação $S_{q}(A+B)$} \\
		\midrule
		Superaditividade & $q < 1$ & $\geq S_{q}(A)+S_{q}(B)$ \\
		Aditividade      & $q = 1$ & $= S_{q}(A)+S_{q}(B)$ \\
		Subaditividade   & $q > 1$ & $\leq S_{q}(A)+S_{q}(B)$ \\
		\bottomrule
	\end{tabular}
	\caption{Propriedades da pseudo-aditividade}
	\label{TAB_sQ}
\end{table}
	O caso \textit{superaditivo} ocorre quando a entropia final é maior (ou igual) que a soma das entropias individuais de cada sistema. Já o caso \text{subaditivo} está relacionado ao fato de que a entropia final dos sistemas conjuntos é menor (ou igual) que a soma das entropias individuais (ver Tabela \ref{TAB_sQ}).
\end{ppdd}

\begin{ppdd}[Positividade]
	Se um determinado estado $n_{0}$ do sistema é totalmente conhecido, então tem-se que $p_{n_{0}}=1$ e $p_{n}=0$ $(\forall n\neq n_{0})$, portanto 
	\begin{align*}
		S_{q}=k\frac{1 - \sum_{n=1}^{W}p_{n}^{q}}{q-1}=k\frac{1 - p_{n_{0}}^{q}}{q-1}=0, \qquad (\forall q\in \mathbb{R}).
	\end{align*}
	Entretanto, para qualquer outro caso, tem-se que $p_{n}<1$, para pelo menos dois valores de $n$ e consequentemente, $1/p_{n}>1$. 
	\begin{align}
		S_{q} = k\frac{1 - \sum_{n=1}^{W} p_{n}^{q}}{q-1} =k\sum_{n=1}^{W} p_{n}{\frac{(1-p_{n}^{q-1})}{q-1}} \qquad (\forall  q\in \mathbb{R}),
	\end{align}
	onde
	\begin{align}
		\ln_{q}p_{i} = \frac{(1-\sum_{n=1}^{W}p_{n}^{q-1})}{q-1}.
	\end{align}
	
	Portanto, dado que $\ln_{q} \frac{1}{p_{n}}>0 ,\forall n$, tem-se que
	\begin{align}
		S_{q} = -k\sum_{i}^{W}p_{i}\ln_{q} p_{i} = k\sum_{i}^{W}p_{i}\ln_{q} \frac{1}{p_{i}} >0 \qquad (\forall  q\in \mathbb{R}).
	\end{align} 
	o que demonstra a positividade da entropia de Tsallis. 
\end{ppdd}

\begin{ppdd}[Expansibilidade]
	Se novos estados possíveis forem adicionados à entropia $S_q$, de modo que suas probabilidades sejam nulas, isto é, $p_{W+1} = \dots = p_{W+N} = 0$, então o valor da entropia permanece inalterado. Portanto, se $p_{k}=0, \; \forall k > W$, então
	\begin{align}\label{expansibilidade_Sq}
		S_{q}(p_{1},p_{2},\dots,p_{W}, p_{W+1},\dots,p_{W+N})=S_{q}(p_{1},p_{2},\dots,p_{W}).
	\end{align}
	
	De fato, se forem adicionados $N$ estados existentes à entropia $S_{q}$, tais que a probabilidade com que eles ocorram seja nula, tem-se para $q>0$ e, por meio da definição fornecida pela equação \eqref{Tsallis_entropia_estados_discretos}
	\begin{align}
		&S_{q}(p_{1},p_{2},\dots,p_{W}, p_{W+1},\dots,p_{W+N})\nonumber\\ &= k\frac{1 - \sum_{i=1}^{W+N} p_{i}^{q}}{q-1} \nonumber\\&= k\frac{1 - \left(\sum_{i=1}^{W} p_{i}^{q}+ \sum_{i=W+1}^{N} 0^{q}\right)}{q-1} \nonumber \\ &= k\frac{1 - \sum_{i=1}^{W} p_{i}^{q}}{q-1} = S_{q}(p_{1},p_{2},\dots,p_{W}).
	\end{align} 
	Entretanto, para $q<0$, deve-se levar em conta o fato de que a soma na entropia $S_{q}$ ocorre apenas em estados com probabilidades positivas \cite{book:166140}. Logo, se há $W+N$ estados existentes, porém somente $W$ estados são possíveis de ocorrer e os $N$ estados restantes têm probabilidade nula de ocorrência, então
	\begin{align}
		S_{q}(p_{1},p_{2},\dots,p_{W},0) &= k\frac{1 - \sum_{i=1}^{W} p_{i}^{q}}{q-1} =S_{q}(p_{1},p_{2},\dots,p_{W}).
	\end{align}
\end{ppdd}

Considere a função definida por 
\begin{align}\label{funcao_arbitraria}
	f(x) \equiv \frac{x(1-x^{q-1})}{q-1} = \frac{x - x^{q}}{q-1}.
\end{align}
A derivada segunda dessa função é dada por
\begin{align}
	\frac{d^{2}f(x)}{dx^{2}} = -qx^{(q-2)}.
\end{align}
Portanto, à luz do Corolário (\ref{corolario_segunda_derivada}), temos que $f(x)$ será côncava se $q>1$ e convexa se $q<1$. Com isso, podemos definir a concavidade e convexidade da entropia de Tsallis.

\begin{ppdd}[Convexidade e concavidade]
	Considere dois conjuntos de probabilidades $\{p_{i}\}$ e $\{\tilde{p}_{i}\}$ que estão associados ao mesmo sistema, com $W$ microestados. E a probabilidade intermediária definida pela equação \eqref{prob_intermediaria}. Como a função dada pela equação \eqref{funcao_arbitraria} possui derivada segunda contínua, positiva para $q<0$ e negativa para $q>0$, tem-se, para $q<0$, segundo a equação \eqref{def_concavidade}
	\begin{align}
		k\frac{p_{i}'(1-p_{i}'^{q-1})}{q-1} < \alpha k\frac{p_{i}(1-p_{i}^{q-1})}{q-1} + (1-\alpha)k\frac{\tilde{p}_{i}(1-\tilde{p}_{i}^{q-1})}{q-1}.
	\end{align} 
	Perfazendo a soma $\sum_{i=1}^{W}$ em ambos os lados da desigualdade, obtém-se 
	\begin{align}\label{def_concavidade_q}
		S_{q}(\{p_{i}'\}) < \alpha S_{q}(\{p_{i}\}) + (1-\alpha)S_{q}(\{\tilde{p}_{i}\}) \qquad (q<0),
	\end{align}
	o que demonstra imediatamente a convexidade da entropia $S_{q}$ para $q<0$. Para obter a concavidade de $S_{q}$ para $q>0$, basta tomar as desigualdades contrárias \cite{book:166140}. 
\end{ppdd}

\subsubsection*{TEOREMA DE SANTOS}

{No trabalho de Santos \cite{dos1997generalization}, é proposto um análogo ao teorema de Shannon para a entropia de Tsallis. Parte-se da hipótese de que existe uma função entropia \( S_{q}(p_{i}) \) que satisfaz as seguintes condições:
	\begin{enumerate}
		\item\label{cond_primeira_santos_teo}  $S_{q}(p_{1},\dots,p_{W})$ é uma função contínua com respeito a todos os seus argumentos;
		\item\label{cond_dois_santos_teo} Para um dado conjunto de $W$ estados equiprováveis, ou seja, $p_{i}=1/W$, é uma função monotonamente crescente de $W$;
		\item\label{cond_tres_santos_teo} Para dois sistemas independentes $A$ e $B$, a entropia do sistema composto $A+B$ satisfaz a relação pseudo-aditividade\footnote{Por simplicidade, está sendo tomado $k=1$, a expressão com a constante pode ser encontrada na equação \eqref{def_pseudo_aditividade}.}
		\begin{align}
			\frac{S_{q}(A+B)}{k} = \frac{S_{q}(A)}{k} + \frac{S_{q}(B)}{k} + (1-q)\frac{S_{q}(A)S_{q}(B)}{k^{2}}
		\end{align}
		\item\label{cond_quatro_santos_teo}  Com $W = W_{L} + W_{M}$, sendo $W_{L}$ dado por $L$ termos e $W_{M}$ dado por $M$ termos, tem-se 
		\begin{align}
			p_{L}\equiv \sum_{i=1}^{W_{L}}p_{i}, 
		\end{align}
		\begin{align}
			p_{M}\equiv \sum_{i=1}^{W_{M}}p_{i}, 
		\end{align}
		e
		\begin{align}
			p_{L} + p_{M} = 1,
		\end{align}
		portanto, a entropia pode ser escrita como
		\begin{align}
			S_{q}(\{p_{i}\}) &= S_{q}(p_{L},p_{M}) + p_{L}^{q}S_{q}\left(\left\{\frac{p_{i}}{p_{L}}\right\}\right)\nonumber \\&+ p_{M}^{q}S_{q}\left(\left\{\frac{p_{i}}{p_{M}}\right\}\right). 
		\end{align}
	\end{enumerate}
	
	O enunciado da condição (\ref{cond_quatro_santos_teo}) proposta por Santos, embora apresente uma estrutura distinta da condição (\ref{terceira_cond_shannon}) de Shannon, possui o mesmo significado. Para compreendê-lo, considere um experimento que consiste em retirar bolas de uma caixa contendo 12 bolas de cores distintas: 4 vermelhas, 3 azuis, 2 verdes e 3 amarelas. Em um sorteio aleatório, a probabilidade de retirar uma bola vermelha é $p_{1} = \frac{1}{3}$, uma bola azul $p_{2} = \frac{1}{4}$, uma bola verde $p_{3} = \frac{1}{6}$ e uma bola amarela $p_{4} = \frac{1}{4}$, conforme ilustrado na Figura~(\ref{fig:diagramaexperimentos3}). Os rótulos 1, 2, 3 e 4 correspondem, respectivamente, às cores vermelha, azul, verde e amarela. A entropia associada a esse sistema é dada por $S_{q}\left(\frac{1}{3}, \frac{1}{4}, \frac{1}{6}, \frac{1}{4}\right)$.

	\begin{figure}[h!]
		\centering
		\includegraphics[width=0.75\linewidth]{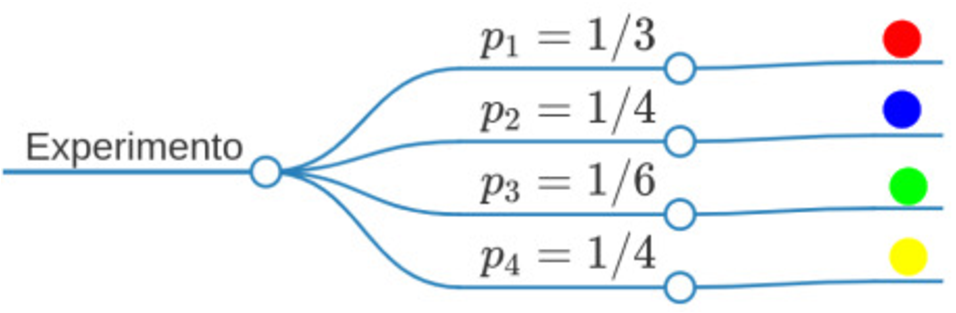}
		\caption{Diagrama ilustrando o experimento com as quatro possibilidades de resultados, cada um associado às probabilidades $p_{1}$, $p_{2}$, $p_{3}$ e $p_{4}$.}
		\label{fig:diagramaexperimentos3}
	\end{figure}

	Considere agora outro experimento com as mesmas 12 bolas do experimento anterior, subdividido em dois sub experimentos consecutivos, como ilustrado na Figura (\ref{fig:diagramaexperimentos4}). O primeiro sub experimento consiste em separar as 12 bolas em dois conjuntos diferentes: um composto por 7 bolas, sendo 4 vermelhas e 3 azuis; o outro por 5 bolas, sendo 2 verdes e 3 amarelas. A probabilidade de selecionar 7 bolas das cores desejadas entre as 12 bolas será $p_{L}=7/12$ e a probabilidade de selecionar as 5 bolas das cores desejadas será de $p_{M}=5/12$.
	
	O segundo sub experimento está condicionado à realização do primeiro. Os conjuntos de 7 e 5 bolas selecionadas são colocados separadamente em uma caixa de sorteio. Do conjunto formado pelas 7 bolas, a probabilidade de retirar uma bola vermelha é $p_{1}/p_{L}=4/7$ e a probabilidade de retirar uma bola azul é $p_{2}/p_{L}=3/7$. Já em relação ao conjunto formado por 5 bolas, a probabilidade de retirar uma bola verde é $p_{3}/p_{M}=2/5$ e de retirar uma bola amarela é $p_{4}/p_{M}=3/5$. As probabilidades $p_{1}/p_{L}$, $p_{2}/p_{L}$, $p_{3}/p_{M}$ e $p_{4}/p_{M}$ são denominadas probabilidades condicionais, pois estão condicionadas ao fato de que no primeiro sub experimento foram selecionados conjuntos de 7 e 5 bolas, cada um com as cores desejadas, como já mencionado anteriormente. 
	
	Nesse experimento formado por dois sub experimentos, a probabilidade de se obter uma bola vermelha (resultado 1) é $p_{1}=1/3$. Isso ocorre porque temos uma probabilidade $p_{L}$ de tirar uma bola vermelha no primeiro sub experimento e uma probabilidade $p_{1}/p_{L}$ de tirar uma bola vermelha no segundo sub experimento. A probabilidade total será fornecida por $p_{L}\cdot\left(\frac{p_{1}}{p_{L}}\right)$. A probabilidade de obter uma bola azul (resultado 2) é $p_{2}=1/4$ e a de obter uma bola verde (resultado 3) é $p_{3}=1/6$. Por fim, a probabilidade de se obter uma bola amarela (resultado 4) é $p_{4}=1/4$. 
	\begin{figure}[h!]
		\centering
		\includegraphics[width=0.95\linewidth]{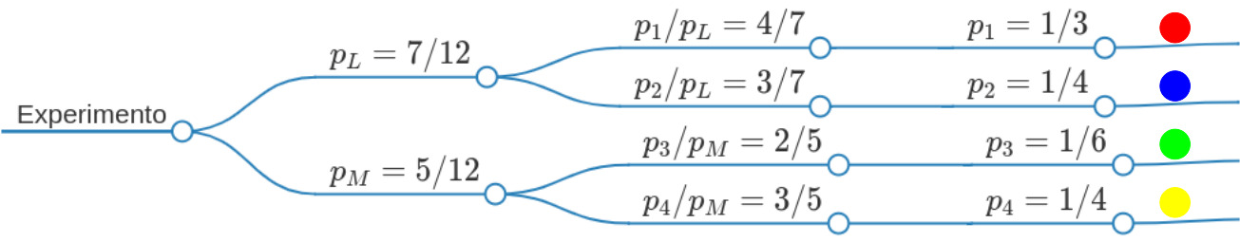}
		\caption{Diagrama ilustrando o experimento. Inicialmente têm-se duas possibilidades de resultados. A escolha com probabilidade $p_{L}=7/12$ se subdivide em outras duas escolhas com probabilidades $p_{1}/p_{L}=4/7$ e $p_{2}/p_{L}=3/7$ que levam aos resultados 1 e 2. Já a escolha $p_{M}=5/12$ se subdivide em duas outras escolhas possíveis com probabilidades $p_{3}/p_{M}=2/5$ e $p_{4}/p_{M}=3/5$, cada uma leva aos resultados 3 e 4, respectivamente.}
		\label{fig:diagramaexperimentos4}
	\end{figure}
Nesse caso, considerando-se as probabilidades dos sub experimentos realizados, a entropia associada ao sistema pode ser escrita como
\begin{align}
	&S_{q}\left(\frac{1}{3},\frac{1}{4},\frac{1}{6},\frac{1}{4}\right) \nonumber\\ &= S_{q}\left(\frac{7}{12},\frac{5}{12}\right) + \left(\frac{7}{12}\right)^{q}S_{q}\left(\frac{4}{7},\frac{3}{7}\right) + \left(\frac{5}{12}\right)^{q}S_{q}\left(\frac{2}{5},\frac{3}{5}\right). 
\end{align}}
{\begin{theorem}[Santos]
		A única função que satisfaz simultaneamente às propriedades (\ref{cond_primeira_santos_teo}), (\ref{cond_dois_santos_teo}) e (\ref{cond_tres_santos_teo}), listadas acima, é a entropia de Tsallis\footnote{Por consistência da notação, está sendo adotado k=1.}
		\begin{align}
			S_{q} = k\frac{1 - \sum_{1}^{W} p_{i}^{q}}{q-1}.
		\end{align}
	\end{theorem}
	\begin{proof}
		Considere um experimento, cuja probabilidade de que um resultado seja obtido é equiprovável e dada por $1/s^{m}$. Podemos subdividir esse experimento em $m$ sub experimentos, com a probabilidade equiprovável $1/s$ de que um determinado resultado seja obtido (ver Figura \ref{fig:umaescolha19}). Desse modo, pode-se obtê-lo utilizando a terceira condição de Santos. Provaremos o primeiro resultado usando o princípio da indução finita (PIF).  Comecemos pelo caso base.
		\begin{figure}[h!]
			\centering
			\includegraphics[width=0.50\linewidth]{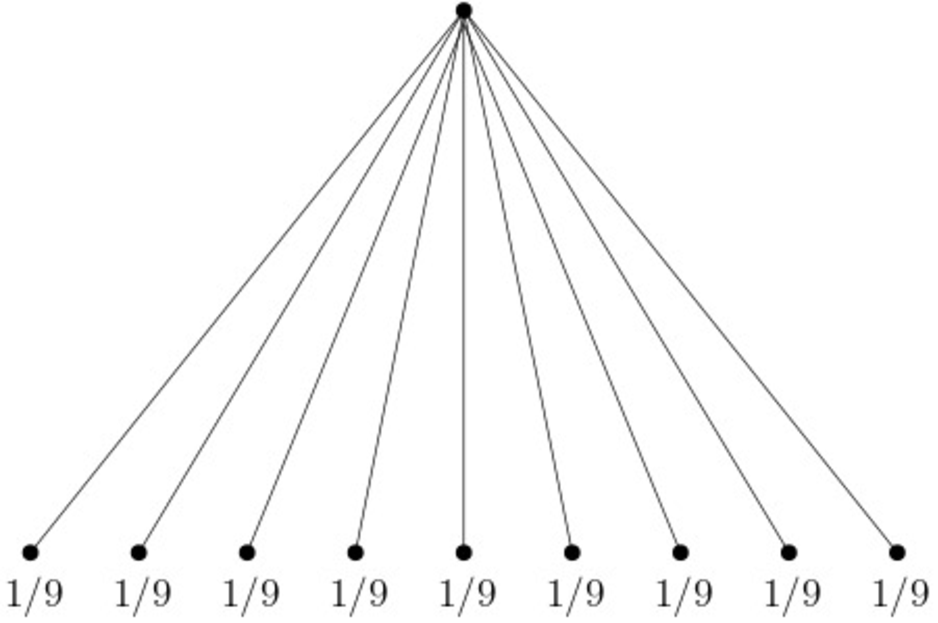}\\
			\textbf{(a)}\\
			\includegraphics[width=0.85\linewidth]{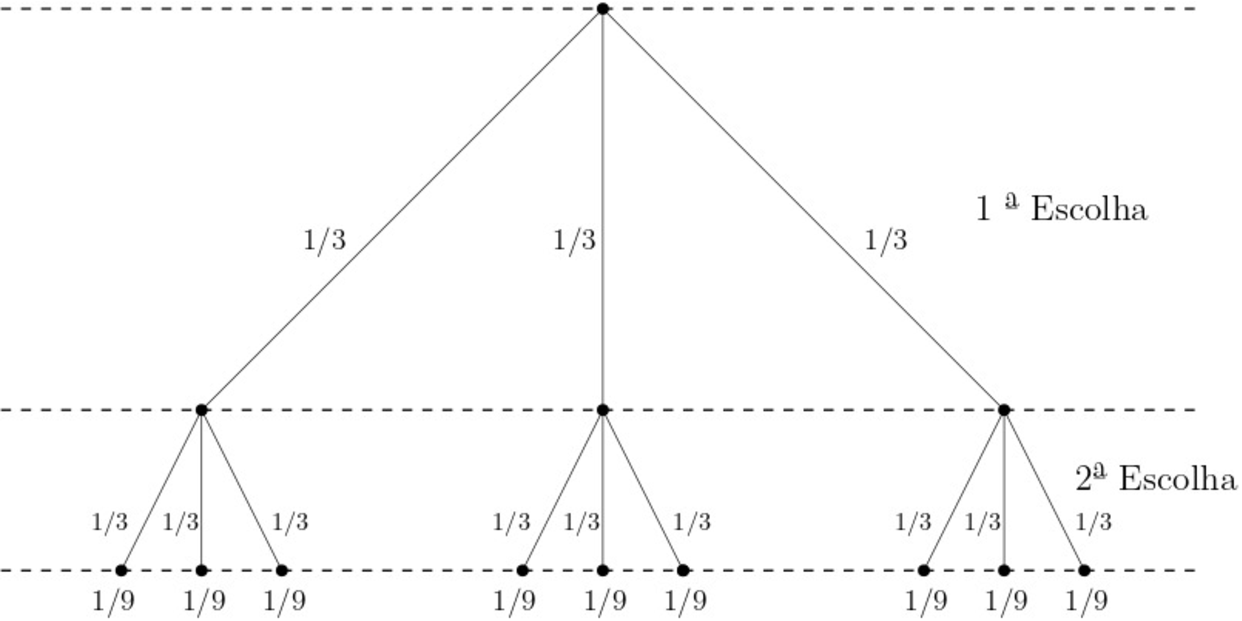}\\
			\textbf{(b)}
\caption{\textbf{(a)} Uma escolha com $9$ possibilidades equiprováveis. Portanto, a probabilidade de que uma das possibilidades seja escolhida é de $1/9$. Nesse caso em particular, $s=3$ e $m=2$. \textbf{(b)} Duas escolhas (m=2), com 3 (s=3) possibilidades equiprováveis cada. Portanto, na primeira escolha, com 3 possibilidades, tem-se, então, que a probabilidade é $1/3$ de que uma das possibilidades seja escolhida. Na segunda, há 3 possibilidades de escolhas, a probabilidade de que uma das possibilidades seja escolhida também é $1/3$. No fim, após as duas escolhas serem realizadas, a probabilidade final de que um dos resultados seja atingido é $1/9$.}
\label{fig:umaescolha19}
		\end{figure}
		\begin{itemize}
			\item Se $m=2$,
			\begin{align}
				S_{q}(s^{2}) &= 2S_{q}(s) + (1-q)S_{q}^{2}(s) \nonumber \\&= \frac{[2(1-q)S_{q}(s) + (1-q)^{2}S_{q}^{2}(s) + 1] - 1}{(1-q)} \nonumber \\ &= \frac{[(1-q)S_{q}(s) + 1]^{2}-1}{(1-q)}.
			\end{align}
			\item  Se $m=3$,
			\begin{align}
				&S_{q}(s^{3}) = S_{q}(s\cdot s^{2}) \nonumber\\&= S_{q}(s) + S_{q}(s^{2}) + (1-q)S_{q}(s)S_{q}(s^{2}) \nonumber\\ 
				&=3S_{q}(s) + (1-q)S_{q}^{2}(s) + 2(1-q)S_{q}^{2}(s) + (1-q)^{2}S_{q}^{3}(s) \nonumber\\
				&= \frac{1 + 3(1-q)S_{q}(s) + 3(1-q)^{2}S_{q}^{2}(s) + (1-q)^{3}S_{q}^{3}(s) - 1}{1-q} \nonumber\\
				&= \frac{[1 + (1-q)S_{q}(s)]^{3} - 1}{1-q}.
			\end{align}
			\item Por hipótese de indução, assuma que, para algum $n$ inteiro,
			\begin{align}\label{suposição_demonstracao}
				S_{q}(s^{n}) = \frac{[1 + (1-q)S_{q}(s)]^{n} - 1}{1-q}.
			\end{align}
			tem-se que para $m=n+1$
			\begin{align*}
				&S(s^{n+1}) = S_q(s^{n}s) \nonumber\\&= S_q(s^n) + S_q(s) +(1-q)S_q(s^n)S_q(s)\nonumber\\
				&= \frac{[1+(1-q)S_q(s)]^n-1}{1-q} + S_q(s) \nonumber\\&+(1-q)\frac{[1+(1-q)S_q(s)]^n-1}{1-q}S(s)\nonumber\\
				&= \frac{[1+(1-q)S_q(s)]^n-1}{1-q} + S_q(s) \nonumber\\&+[1+(1-q)S_q(s)]^nS_q(s)-S_q(s)\nonumber\\
				&= \frac{[1+(1-q)S_q(s)]^n-1}{1-q}  +[1+(1-q)S_q(s)]^nS_q(s)\nonumber\\
				&= \frac{[1+(1-q)S_q(s)]^n-1}{1-q}  \nonumber\\&+\frac{(1-q)}{1-q}[1+(1-q)S_q(s)]^nS_q(s)\nonumber\\
				&= \frac{[1+(1-q)S_q(s)]^n}{1-q}\left[1+ (1-q)S_q(s)\right] -\frac{1}{1-q}\nonumber\\
				&= \frac{[1+(1-q)S_q(s)]^{n+1}}{1-q}-\frac{1}{1-q}\nonumber\\
				&= \frac{[1+(1-q)S_q(s)]^{n+1}-1}{1-q}.\nonumber\\
			\end{align*}
		\end{itemize}
		\vspace{-1.1cm} Logo, $S_{q}(s^{m}) = \frac{[1 + (1-q)S_{q}(s)]^{m} - 1}{1-q}$. No limite $q\to 1$, a condição de Boltzmann-Gibbs é recuperada
		\begin{align}
			S_{1}(s^{m}) = \lim_{q\rightarrow 1}\frac{[1 + (1-q)S_{q}(s)]^{m} - 1}{1-q}=mS_{1}(s).
		\end{align}
		Para um par $(m,s)$ suficientemente grande, é possível encontrar um par de números inteiros $(t,n)$, tal que
		\begin{align}\label{eq_dos_pares_sm_tn}
			s^{m} \leq t^{n} \leq s^{m+1}.
		\end{align} 
		Usando a condição \eqref{cond_dois_santos_teo}, obtém-se, para qualquer valor de $q$
		\begin{align}\label{eq_sq_m}
			S_{q}(s^{m}) \leq S_{q}(t^{n}) \leq S_{q}(s^{m+1}).
		\end{align}
		Utilizando a equação \eqref{suposição_demonstracao} para $q<1$, tem-se
		\begin{align}
			\left[1+ (1-q)\frac{S_{q}(s)}{k}\right]^{m} &\leq \left[1 + (1-q)\frac{S_{q}(t)}{k}\right]^{n}\nonumber\\ &\leq \left[1 + (1-q)\frac{S_{q}(s)}{k}\right]^{m+1},
		\end{align}
		tomando o logaritmo da desigualdade acima temos
		\begin{align}
			m\ln\left[1+(1-q)\frac{S_{q}(s)}{k}\right] &\leq n\ln\left[1+(1-q)\frac{S_{q}(t)}{k}\right]\nonumber \\ &\leq (m+1)\ln\left[1+(1-q)\frac{S_{q}(s)}{k}\right],
		\end{align}
		dividindo por $n\ln\left[1+(1-q)\frac{S_{q}(t)}{k}\right]$
		\begin{align}
			\frac{m}{n} \leq \frac{\ln \left[1 + (1-q)\frac{S_{q}(t)}{k}\right]}{\ln \left[1 + (1-q)\frac{S_{q}(s)}{k}\right]} \leq \frac{m}{n} + \frac{1}{n}.
		\end{align}
		Somando $-m/n$ em ambos os lados da desigualdade, resultado em
		\begin{align}
			0\leq \frac{\ln \left[1 + (1-q)\frac{S_{q}(t)}{k}\right]}{\ln \left[1 + (1-q)\frac{S_{q}(s)}{k}\right]} - \frac{m}{n}\leq \frac{1}{n},
		\end{align}
		ou equivalentemente 
		\begin{align}\label{ln_eq_st_total}
			\left| 	\frac{m}{n} - \frac{\ln \left[1 + (1-q)\frac{S_{q}(t)}{k}\right]}{\ln \left[1 + (1-q)\frac{S_{q}(s)}{k}\right]} \right| \leq \frac{1}{n} \equiv \varepsilon.
		\end{align}
		Da equação \eqref{eq_dos_pares_sm_tn}
		\begin{align}
			\frac{m}{n} \leq \frac{\ln t}{\ln s} \leq \frac{m}{n} + \frac{1}{n},
		\end{align}
		e, de modo análogo
		\begin{align}\label{qe_lnts}
			\left|\frac{m}{n} - \frac{\ln t}{\ln s}\right| \leq \frac{1}{n} \equiv \varepsilon.
		\end{align}
		Usando as equações \eqref{ln_eq_st_total} e \eqref{qe_lnts} combinadas
		\begin{align}
			&\left|\frac{\ln t}{\ln s} - \cancel{\frac{m}{n}} + \cancel{\frac{m}{n}}+ \frac{\ln \left[1 + (1-q)\frac{S_{q}(t)}{k}\right]}{\ln \left[1 + (1-q)\frac{S_{q}(s)}{k}\right]}\right|\nonumber \\ &\leq \left|\frac{m}{n} - \frac{\ln t}{\ln s}\right| + \left|\frac{m}{n}-\frac{\ln \left[1 + (1-q)\frac{S_{q}(t)}{k}\right]}{\ln \left[1 + (1-q)\frac{S_{q}(s)}{k}\right]}\right| \leq 2\varepsilon,
		\end{align}
		concluí-se
		\begin{align}
			\left|\frac{\ln t}{\ln s} - \frac{\ln\left[1 + (1-q)\frac{S_{q}(t)}{k}\right]}{\ln\left[1+(1-q)\frac{S_{q}(s)}{k}\right]}\right| \leq 2\varepsilon. 
		\end{align}
		No limite $\varepsilon\rightarrow 0$, obtém-se
		\begin{align}
			\frac{\ln\left[1+(1-q)\frac{S_{q}(s)}{k}\right]}{\ln s} = \frac{\ln\left[1+(1-q)\frac{S_{q}(t)}{k}\right]}{\ln t} = p(q).
		\end{align}
		Logo, a função $S_{q}(t)$ toma a forma de
		\begin{align}\label{passo_intermediario_demonstracao}
			S_{q}(t) = k\frac{t^{p(q)} -1}{1-q}.
		\end{align}
		Considere agora uma escolha de $W$ partições, cada uma com probabilidade 
		\begin{align}\label{ppdd_particoes}
			p_{i}  = \frac{n_{i}}{\sum_{j=1}^{W} n_{j}},
		\end{align}
		onde $n_{i}$ é o número de possibilidades na $i-$ésima partição, cada uma com probabilidade igual. Por exemplo, se uma caixa com dez bolinhas diferentes, como ilustrado na Figura (\ref{fig:imagemcaixaspossibilidades}) for dividida em $W=4$ partições, representadas pelas cores amarela, verde, vermelha e azul. A probabilidade de que uma das cores saia é dada pela equação \eqref{ppdd_particoes}. Onde $n_{1}=4$, $n_{2}=3$, $n_{3}=2$ e $n_{4}=1$. Nesse esquema, a probabilidade de que eu escolha a partição amarela é de $2/5$ e que nessa partição eu tire a bolinha azul é $1/4$. Logo a probabilidade que, em uma escolha, eu tire a bolinha azul é dada pelo produto das duas probabilidades, que é igual a $1/10$.
		\begin{figure}[h!]
			\centering
			\includegraphics[width=0.4\linewidth]{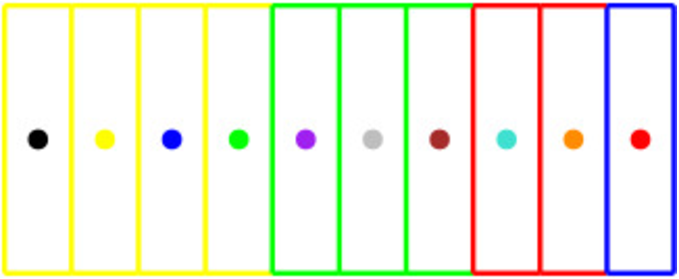}
			\caption{Uma caixa com $W=4$ partições. Cada cor está associada a uma partição diferente. A partição amarela tem $n_{1}=4$ possibilidades em uma escolha, a verde $n_{2}=3$, a vermelha $n_{3}=2$ e a azul $n_{4}=1$ possibilidades.}
			\label{fig:imagemcaixaspossibilidades}
		\end{figure}
		
		Usando a condição \eqref{cond_quatro_santos_teo} de Santos, tem-se
		\begin{align}
			S_{q}\left(\left\{\frac{1}{\sum_{i}^{W}n_{i}}\right\} \right) = S_{q}(p_{i}) + \sum_{i=1}^{W} p_{i}^{q} S_{q}\left(\left\{\frac{1}{n_{i}}\right\}\right)
		\end{align}
		por meio da equação \eqref{passo_intermediario_demonstracao}, pode-se escrever
		\begin{align}
			k\frac{\left(\sum_{i=1}^{W}n_{i}\right)^{p}-1}{1-q} = S(p_{i}) + k\sum_{i=1}^{W} p_{i}^{q}\left(\frac{n_{i}^{p}-1}{1-q}\right).
		\end{align}
		À luz do exemplo anterior, podemos interpretar que a entropia de tirar uma bolinha de determinada cor é dada pela soma da entropia de tirar uma determinada partição em uma escolha mais a soma ponderada das entropias de escolher uma determinada cor em uma partição. Então
		\begin{align}
			S_{q}(p_{i}) = \frac{k}{1-q}\left\{\left(\sum_{i=1}^{W}n_{i}\right)^{p} - 1 + \sum_{i=1}^{W}p_{i}^{q} - \sum_{i=1}^{W}p_{i}^{q}n_{i}^{p}\right\},
		\end{align}
		mas da equação \eqref{ppdd_particoes}, temos
		\begin{align}
			n_{i}^{p} = p_{i}^{p}\left(\sum_{i=1}^{W}n_{i}\right)^{p},
		\end{align}
		logo
		\begin{align}
			S_{q}(p_{i}) =& \frac{k}{1-q}\left\{\left(\sum_{i=1}^{W}n_{i}\right)^{p} - 1 + \sum_{i=1}^{W}p_{i}^{q}\right.\nonumber\\&\left. - \sum_{j=1}^{W}p_{j}^{q}p_{j}^{p}\left(\sum_{i=1}^{W}n_{i}\right)^{p}\right\}.
		\end{align}
		Para satisfazer a condição \eqref{cond_primeira_santos_teo}, deve-se observar que
		\begin{align}
			p = 1-q.
		\end{align}
		Portanto
		\begin{align}
			S_{q}(p_{i}) = k\frac{1-\sum_{i=1}^{W}p_{i}^{q}}{q-1}, 
		\end{align}
		e para $t$ escolhas equiprováveis
		\begin{align}
			S_{q}(t) = k\frac{t^{1-q}-1}{q-1}.
		\end{align}
	\end{proof}
	
	Portanto, assim como teorema de Shannon garante a unicidade da entropia de Boltzmann-Gibbs sob condições razoáveis, o teorema de Santos garante a unicidade da entropia de Tsallis sob condições análogas.

	\subsection{SOBRE ADITIVIDADE E EXTENSIVIDADE}

É usual referir-se à mecânica estatística não aditiva de Tsallis como mecânica estatística não extensiva \cite{tsallis2004nonextensive,nunes2002nonextensive,boon2005special}. Isso porque a entropia resultante de dois subsistemas não corresponde à soma das entropias individuais para $q\neq 1$. A mecânica estatística de Boltzmann-Gibbs, em contrapartida, cuja entropia resultante da composição de dois subsistemas é a soma das entropias individuais, é considerada extensiva. Nesse contexto, o conceito de extensividade parece estar associado ao fato de a entropia ser ou não aditiva \cite{landsberg1999entropies}, o que pode nos conduzir a uma ideia equivocada desse conceito.

A extensividade é um conceito termodinâmico e fundamental para estabelecer a estrutura das transformações de Legendre em termodinâmica \cite{tsallis2014introduction}. Sua definição, contudo, não depende de a entropia ser composta de forma aditiva ou não. A aditividade, por sua vez, depende apenas de como a entropia é definida em termos probabilísticos. De forma precisa, pode-se definir\cite{book:166140}

\begin{definition}[Extensividade]
	Seja $N$ o número de elementos que compõem um sistema ou subsistema, sejam eles probabilisticamente independentes ou não. A entropia desse sistema ou subsistema é dita extensiva se for proporcional a $N$.
\end{definition}

Portanto, uma entropia pode ser aditiva e não extensiva para um determinado sistema termodinâmico, como é o caso da entropia de Boltzmann-Gibbs para sistemas fortemente correlacionados. Por outro lado, uma entropia não aditiva pode ser extensiva para um dado sistema, a depender de um valor específico de $q$.

	\subsection{ENSEMBLE MICROCANÔNICO E ENSEMBLE CANÔNICO}
	
	\textcolor{black}{Assim como os ensembles microcanônico e canônico podem ser analisados a partir da entropia de Boltzmann-Gibbs, eles também podem ser descritos por meio da entropia de Tsallis, quando se trata de sistemas compatíveis com essa generalização estatística, conforme discutido em \cite{plastino1999tsallis}. A seguir, apresentamos uma derivação detalhada da entropia no ensemble microcanônico e da função de partição no ensemble canônico nesse contexto.}

	\begin{proposition}[Ensemble microcanônico]
		Em um sistema físico isolado a entropia de Tsallis \eqref{Tsallis_entropia_estados_discretos}, sob a condição de normalização da probabilidade \eqref{condicao_de_normalizacao}, assume seu valor extremizado dado pela equação \eqref{entropia_tsallis_equiprobabilistica} e a probabilidade é dada por $p_{i}=1/W$, onde $W$ é o número de estados acessíveis ao sistema.
	\end{proposition} 
	
	\begin{proof}
		Utilizaremos o método de multiplicadores de Lagrange para extremizar a entropia de Tsallis sob a condição de normalização. Para isso, definimos a função
		\begin{align}
			\mathcal{L} =  k\frac{1-\sum_{i=1}^{W}p_{i}^{q}}{q-1}+ \lambda\left(1 - \sum_{i=1}^{W} p_{i}\right),
		\end{align}
		onde $\lambda$ é o multiplicador de Lagrange.
		Variando a função acima, temos que
		\begin{align}
			\delta\mathcal{L} = \left(-k\frac{q}{q-1}p_{i}^{q-1} -\lambda\right)\delta p_{i} =0.
		\end{align}
		
		Desse modo, podemos escrever a probabilidade como
		\begin{align}\label{ppdd_multiplicador_lagrange_tsallis}
			p_{i} = \left[-\frac{\lambda (q-1)}{kq}\right]^{1/(q-1)}.
		\end{align}
		
		Substituindo a probabilidade na condição de normalização dada pela equação \eqref{condicao_de_normalizacao}, temos
		\begin{align}
			\sum_{i=1}^{W}  \left[-\frac{\lambda (q-1)}{kq}\right]^{1/(q-1)} = W \left[-\frac{\lambda (q-1)}{kq}\right]^{1/(q-1)}=1.
		\end{align}
		Portanto, o multiplicador de Lagrange é dado por 
		\begin{align}
			\lambda = -\frac{kq}{(q-1)W^{q-1}},
		\end{align}
		e a probabilidade encontrada na equação \eqref{ppdd_multiplicador_lagrange_tsallis} se reduz à
		\begin{align}
			p_{i} = \frac{1}{W}.
		\end{align}
		Dessa maneira, podemos escrever a entropia, no caso em que ela é extremizada como
		\begin{align}
			S_{q}(W) = k\frac{W^{q-1}-1}{1-q}.
		\end{align}
	\end{proof}
	
	Note que, diferentemente de Boltzmann-Gibbs, a entropia de Tsallis quando extremizada sob a condição de normalização, assume valor máximo ou mínimo a depender do valor de $q$. Basta olharmos que a a segunda derivada é dada por
	\begin{align}
		\left.\frac{\partial^{2} S_{q}(W)}{\partial {p_{i}}^{2}}\right|_{p_{i}=1/W} = -k q W^{2-q}.
	\end{align} 
	Para $q>0$ é máxima e para $q<0$ é mínima.

	Definindo a energia interna de um sistema que satisfaz a mecânica estatística não aditiva como \cite{EMFCurado1991}
	\begin{align}\label{energia_media_ensemble_canonico_q}
		\langle E_{q}\rangle = \sum_{i=1}^{W} p^{q}_{i}\varepsilon_{i},
	\end{align}
	podemos introduzir a ensemble canônico e a função de partição para mecânica estatística baseada na entropia de Tsallis.
	
	\begin{proposition}[Ensemble canônico]
		Um sistema físico em contato com um reservatório térmico, que satisfaz entropia de Tsallis \eqref{Tsallis_entropia_estados_discretos}, sob a condição de normalização \eqref{condicao_de_normalizacao}, com energia média dada pela equação \eqref{energia_media_ensemble_canonico_q} e à luz do princípio da máxima entropia, é caracterizado pela função de partição dada por
		\begin{align}
			Z_{q} = \sum_{i=1}^{W} \left[1 - (1-q)\frac{\beta}{k}\varepsilon_{i}\right]^{1/(1-q)}.
		\end{align}
		
		\begin{proof}
			Para encontrar a função de partição, devemos extremizar a função entropia pelo método de multiplicadores de Lagrange utilizando como restrições a condição de normalização \eqref{condicao_de_normalizacao} e a energia interna média \eqref{energia_media_ensemble_canonico_q}. Para isso, definimos a seguinte função

			\begin{align}
				\mathcal{L} =  k\frac{1-\sum_{i=1}^{W}p_{i}^{q}}{q-1} + \lambda\sum_{i=1}^{W} p_{i} - \beta \sum_{i=1}^{W} p^{q}_{i}\varepsilon_{i}.
			\end{align}
			Variando a função acima, obtemos
			\begin{align}
				\delta\mathcal{L} = \sum_{i=1}^{W}\left(-\frac{kq}{(q-1)}p_{i}^{(q-1)} +\lambda - \beta q \varepsilon_{i} p_{i}^{(q-1)} \right)\delta p_{i}=0.
			\end{align}
			Portanto, $\forall i$ temos que
			\begin{align}
				p_{i}^{(q-1)}\left(\frac{kq}{(q-1)} + q\beta\varepsilon_{i}\right) = \lambda,
			\end{align}
			\begin{align}\label{ppdd_ensemble_canonico}
				p_{i} = \left[\frac{(q-1)\lambda}{kq}\right]^{1/(q-1)}\frac{1}{\left[1+ (q-1)\frac{\beta}{k}\varepsilon_{i}\right]^{1/(q-1)}}.
			\end{align}
			Usando a condição de normalização,
			\begin{align}
				\sum_{i=1}^{W} p_{i} =  \sum_{i=1}^{W} \left[\frac{(q-1)\lambda}{q}\right]^{1/(q-1)}\frac{1}{\left[1+ (q-1)\frac{\beta}{k}\varepsilon_{i}\right]^{1/(q-1)}} =1 ,
			\end{align}
			\begin{align}
				\lambda^{1/(q-1)} = \frac{1}{\left[\frac{(q-1)}{q}\right]^{1/(q-1)}\sum_{i=1}^{W}[1 + (q-1)\frac{\beta}{k}\varepsilon_{i}]^{1/(1-q)}}.
			\end{align}
			Substituindo o valor de $\lambda$ na equação \eqref{ppdd_ensemble_canonico}, temos
			\begin{align}\label{ppdd_canonico_q}
				p_{i} = \frac{[1 - (1-q)\frac{\beta}{k}\varepsilon_{i}]^{1/(1-q)}}{Z_{q}}=\frac{e_{q}^{\frac{\beta\varepsilon_{i}}{k}}}{Z_{q}},
			\end{align}
			onde 
			\begin{align}
				Z_{q} = \sum_{i=1}^{W}\left[1 - (1-q)\frac{\beta}{k}\varepsilon_{i}\right]^{1/(1-q)} = \sum_{i=1}^{W} e_{q}^{\frac{\beta\varepsilon_{i}}{k}},
			\end{align}
			é a função de partição.
		\end{proof}
		
		Dessa maneira, conhecida a função de partição do sistema descrito pela entropia de Tsallis, pode-se então encontrar sua entropia, dado que a partir da função de partição pode-se obter as probabilidades.
	\end{proposition}

\subsection{APLICAÇÕES}\label{aplicacoes}

Os formalismos de mecânica estatística discutidos neste trabalho constituem um arcabouço teórico de amplo alcance, com aplicações que se estendem da física à linguística, passando pela computação, biologia, economia e diversas outras áreas do conhecimento.

A abordagem convencional, baseada na mecânica estatística de Boltzmann–Gibbs e usualmente abordada em cursos de graduação, oferece um repertório clássico de aplicações em física. Incluindo o gás ideal clássico e o gás de Fermi, o condensado de Bose–Einstein e o sólido de Einstein, ambos facilmente encontradas em referências canônicas \cite{book:18204,book:1322644,reif2009fundamentals}. Para além das aplicações tradicionais à física, esse formalismo tem sido amplamente explorado em pesquisas interdisciplinares, por exemplo:

\begin{itemize}
\item em linguística, para modelar a aquisição e a evolução das línguas ao longo do tempo \cite{cassandro1999statistical,kosmidis2006statistical};
\item em biologia, por meio do princípio da máxima entropia, para investigar a diversidade de espécies em ecossistemas \cite{banavar2010applications};
\item em economia, para analisar dinâmicas de mercado financeiro, dando origem à econofísica \cite{yakovenko2009econophysics,ausloos2000statistical,pereira2017econophysics}.
\end{itemize}

Essas aplicações reforçam a posição da mecânica estatística de Boltzmann–Gibbs como um dos importantes campos de aplicações modernas da física.

De forma análoga, a mecânica estatística de Tsallis tem ocupado papel de destaque na física contemporânea, impulsionando avanços como o cálculo do calor específico do átomo de hidrogênio para $q<1$ \cite{PhysRevE.51.6247} e a análise de mapas logísticos \cite{tsallis_2021,tirnakli2016standard}. Outras aplicações relevantes incluem:

\begin{itemize}
\item a investigação sobre buracos negros a partir da entropia de Tsallis permitiu explorar a estabilidade termodinâmica de diferentes soluções sob uma outra perspectiva. Observou-se que buracos negros de Schwarzschild, por exemplo, embora sejam termodinamicamente instáveis no contexto da entropia de Bekenstein-Hawking, tornam-se estáveis quando analisados sob a estatística de Tsallis, desde que o parâmetro $q$ satisfaça $q < 1$.\cite{abreu2021nature,chunaksorn2025black,mejrhit2019thermodynamics};

\item a análise das curvas de luz em raios X de sistemas astrofísicos revelou que elas seguem uma distribuição $q$-exponencial, com um valor médio do parâmetro $q$ estimado em $q = 1{,}418 \pm 0{,}007$ \cite{rosa2013nonextensivity}. Nesse mesmo trabalho, demonstrou-se que quanto mais entrópico é o sistema, menor tende a ser o valor do parâmetro $q$. Já em \cite{carvalho2007radial}, foi mostrado que a distribuição da velocidade radial em aglomerados estelares também obedece à estatística $q$-exponencial. Além disso, observou-se uma correlação entre o parâmetro $q$ e a idade de aglomerados com mais de $10^9$ anos: quanto mais antigo o aglomerado e mais distante do centro galáctico, menos a distribuição de velocidades se assemelha à distribuição de Maxwell;

\item na área da geofísica, o formalismo estatístico de Tsallis foi utilizado para analisar a atividade sísmica na Grécia entre os anos de 1980 e 2010 \cite{michas2013non}. Constatou-se que a variação do parâmetro $q$, no intervalo entre $1,26$ e $1,54$, está correlacionada ao acúmulo de energia sísmica. Observou-se também que, à medida que o sistema se afastava do equilíbrio, tanto o valor de $q$ quanto a liberação de energia aumentavam, indicando a maior probabilidade de ocorrência de terremotos de maior magnitude \cite{sigalotti2023tsallis}. Outro resultado reportado mostra que a distância tridimensional entre terremotos sucessivos é bem descrita por uma distribuição $q$-exponencial com $q<1$ \cite{abe2003law};

	\item em biologia, a distribuição de velocidades de células \textit{Hydra} ajusta-se à distribuição q-exponencial com $q\approx 1,5$ \cite{upadhyaya2001anomalous}. A estatística de Tsallis também foi utilizada para analisar séries temporais de doenças neurodegenerativas. Os resultados indicam que o valor estimado do parâmetro $q$ pode ser utilizado como um dos possíveis marcadores biológicos, a fim de melhorar as previsões de crises epilépticas. \cite{iliopoulos2016tsallis}.
\end{itemize}

Para uma visão mais abrangente das aplicações físicas e interdisciplinares da entropia de Tsallis, o leitor pode consultar obras de referência \cite{book:166140,gell2004nonextensive,tsallis2011nonadditive,tsallis2009computational,tsallis1999nonextensive}. As aplicações destacadas, somadas às evidências experimentais que as sustentam \cite{boghosian1996thermodynamic}, têm consolidado a mecânica estatística de Tsallis como uma generalização consistente da abordagem de Boltzmann-Gibbs, permitindo a descrição de fenômenos que escapam ao escopo do formalismo tradicional.

	\subsection{ENTROPIA DE RÉNYI}
	Propostas alternativas de funcionais entrópicos, além dos supracitados, foram feitas. Dentre elas está a entropia de Rényi \cite{renyi1970probability} definida como
	\begin{align}\label{renyi_entropy_def}
		S_{q}^{R} \equiv k\frac{\ln \sum^{W}_{i=1}p_{i}^{q}}{1-q} .
	\end{align}
	Na literatura, diz-se com frequência que a entropia de Rényi é obtida via \textit{abordagem do logaritmo formal} da entropia de Tsallis $S_{q}$, como pode ser observado na equação \eqref{logaritmo_formal_tsallis} \cite{Czinner:2015eyk} e que resulta em
	\begin{align}
		S_{q}^{R} (p_{i}) = k \frac{\ln\left[1+(1-q)\frac{S_{q}}{k}\right]}{1-q}.
	\end{align}
	
	A entropia de Rényi tem a vantagem de ser aditiva, ou seja
	\begin{align}
		S_{q}^{R}(A+B) = S_{q}^{R}(A) + S_{q}^{R}(B),
	\end{align}
	esse fato permite que a lei zero da termodinâmica seja satisfeita, o que não acontece com a entropia de Tsallis \cite{ccimdiker2023equilibrium}. No entanto, ela é côncava apenas para $0<q\leq1$ \cite{boon2005special}, isso porque as funções $\ln(x)$ e $x^{q}$ são côncavas nesse intervalo \cite{JIZBA200417}. Viola para $q\neq 1$ a robustez experimental, a produção de entropia por unidade de tempo e a própria concavidade fora do intervalo já mencionado\footnote{Para entender o que é \textit{robustez experimental} e \textit{produção de entropia por unidade de tempo}, veja \cite{tsallis2003introduction} e \cite{book:166140}.}. Entretanto, no limite em que $q\rightarrow 1$, tem-se, utilizando a relação
	\begin{align}
		\ln(1+x) \approx x,
	\end{align}
	na equação \eqref{renyi_entropy_def},
	\begin{align}
		S_{q=1}^{R} = \frac{(1-q)S_{q=1}}{1-q} = S_{1} = S_{BG}.
	\end{align}
	Logo a entropia de Rényi, possui a entropia de Boltzmann-Gibbs \eqref{entropia_BG_estados_discretos} como caso particular, de modo análogo como a entropia de Tsallis.

	A entropia de Rényi, contudo, não desempenha um papel tão importante quanto as entropias de Boltzmann-Gibbs e de Tsallis na mecânica estatística. Isso se deve, sobretudo, à sua incompatibilidade com uma generalização consistente das leis termodinâmicas, o que limita sua aplicabilidade a sistemas compostos por muitas partículas ou que envolvam interações complexas \cite{boon2005special}.

\section{Considerações finais}
	
No presente trabalho foi realizada uma revisão detalhada de aspectos formais das diferentes formulações para a entropia, com ênfase nas entropias de Boltzmann-Gibbs e de Tsallis. A entropia de Tsallis tem propriedades mais gerais quando comparada à entropia de Boltzmann-Gibbs e, nesse sentido, configura-se, ao menos formalmente, como uma candidata consistente à sua extensão. No entanto, sob a perspectiva termodinâmica, a entropia de Tsallis apresenta limitações, notadamente a violação da lei zero da termodinâmica, em decorrência de sua não aditividade \cite{nauenberg2003critique}. Ainda assim, seu estudo abre caminhos relevantes para a investigação de sistemas fora do equilíbrio e com correlações de longo alcance.

A discussão desenvolvida neste artigo não apenas esclarece e demonstra propriedades matemáticas das funções entrópicas, como também introduz fundamentos da mecânica estatística não extensiva, proporcionando uma visão integrada das diferentes formulações entrópicas. Espera-se que esta análise, apresentada em linguagem acessível a estudantes de graduação, contribua para o aprofundamento da reflexão e da investigação sobre as diversas formulações entrópicas e os fundamentos da mecânica estatística não extensiva e sistemas complexos.

\section{Agradecimentos}

KSA agradece o apoio da Coordenação de Aperfeiçoamento de Pessoal de Nível Superior Brasil (CAPES) - Código de Financiamento 001 pelo apoio financeiro. RTC agradece ao Conselho Nacional de Desenvolvimento Científico e Tecnológico - CNPq (Processo 401567/2023-0), pelo suporte financeiro parcial.

\end{document}